\documentclass{LMCS}

\usepackage{verbatim}

\usepackage{graphicx}
\usepackage{hyperref}

\newcommand{\BC}{\mathbb C}

\newcommand{\BA}{\mathbb A}
\newcommand{\BB}{\mathbb B}

\newcommand{\BT}{\mathbb T}
\newcommand{\BL}{\mathbb L}
\newcommand{\BX}{\mathbb X}
\newcommand{\ig}{\mbox{\rm Inc} }
\newcommand{\Block}{\mbox{\rm Block} }

\newtheorem{theorem}{Theorem}[section]
\newtheorem{lemma}[theorem]{Lemma}
\newtheorem{corollary}[theorem]{Corollary}

\def\doi{3 (4:6) 2007}
\lmcsheading%
{\doi}
{1--22}
{}
{}
{Mar.~\phantom{0}2, 2007}
{Nov.~\phantom{0}6, 2007}
{}

\begin{document}

\title[A Characterisation of
first order CSP's]{A Characterisation of First-Order Constraint
Satisfaction Problems\rsuper *}

\author[B. Larose]{Benoit Larose\rsuper a}       
\address{{\lsuper a}Department of Mathematics and Statistics \\
 Concordia University \\
 1455 de Maisonneuve West\\
Montr\'eal, Qc \\
Canada, H3G 1M8}     
\email{larose@mathstat.concordia.ca}  
\thanks{{\lsuper a}Research partially supported by NSERC, FQRNT and CRM}   

\author[C.~Loten]{Cynthia Loten\rsuper b}     
\address{{\lsuper b}Department of Mathematics and Statistics \\
University College of the Fraser Valley \\ 33844 King Rd
Abbotsford, BC Canada V2S 7M8 }
\email{cindy.loten@shaw.ca}  

\author[C.~Tardif]{Claude Tardif\rsuper c}     
\address{{\lsuper c}Department of Mathematics and Computer Science\\
Royal Military College of Canada \\
PO Box 17000 Station ``Forces'' \\
Kingston, Ontario\\
Canada, K7K 7B4 }   
\email{Claude.Tardif@rmc.ca }  
\thanks{{\lsuper c}Research partially supported by NSERC and ARP}    

\titlecomment{{\lsuper *}A short version of this paper appeared in the Proceedings of the 21st
Symposium on Logic in Computer Science (LICS 2006).}

\begin{abstract}
  \noindent We describe simple algebraic and combinatorial characterisations of
  finite relational core structures admitting
finitely many obstructions. As a consequence, we show that it is
decidable to determine whether a constraint satisfaction problem
is first-order definable: we show the general problem to be {\bf
NP}-complete, and give a polynomial-time algorithm in the case of
cores. A slight modification of this algorithm provides, for
first-order definable CSP's, a simple poly-time algorithm to
produce a solution when one exists. As an application of our
algebraic characterisation of first order CSP's, we describe a
large family of {\bf L}-complete CSP's.
\end{abstract}

\maketitle

\section{Introduction}

 The Constraint Satisfaction
Problem (CSP) consists of determining, given a finite set of
variables with constraints on these, whether there exists an
assignment of values  to these variables that satisfies all the
given constraints. The great flexibility of this framework has
made the CSP the focus of a great deal of attention from
researchers in various fields (see for instance the recent survey
\cite{CohenJ06}). In general the problem is {\bf NP}-complete, but
restricting the type of constraint relations involved may yield
tractable problems. In fact, Schaefer \cite{sch} and more recently
\cite{ABISV05} have completely classified the complexity of
Boolean CSP's and from their work it follows that Boolean CSP's
are either trivial, first-order definable, or  complete (under
$AC^0$ reductions) for one of the following standard classes of
problems: {\bf L}, {\bf NL}, {\bf P}, $\oplus${\bf  L} and {\bf
NP}. One of the outstanding problems in the field is the so-called
{\em dichotomy conjecture} \cite{fedvar1} that states that every
CSP should be either in {\bf P} or {\bf NP}-complete.

In this paper we adopt the convenient point of view offered in
\cite{fedvar2} where CSP's are viewed as homomorphism problems
with a fixed target. In other words, if $\BA$ is a finite
relational structure, then $\BA$-CSP consists of all structures
that admit a homomorphism to $\BA$. Viewed this way, it becomes
natural to ask which CSP's can be described in various logics. For
instance in \cite{LT07}, the result of Allender et al. mentioned
earlier is given a descriptive complexity analog, whereby it is
shown that Boolean CSP's that lie in the classes {\bf L} and {\bf
NL} are precisely those whose complement is describable in
symmetric and linear Datalog respectively. Arguably the simplest
CSP's (other than trivial ones) are those whose members are
describable by a first-order sentence. A very natural question in
the vein of the dichotomy conjecture is then the following: can we
determine (easily) from the constraint relations whether a given
CSP is first-order definable ? Related questions for Datalog and
its restrictions remain open \cite{Dalmau05}, \cite{fedvar2}. An
important first step in this direction is Atserias' result
\cite{ats} proving that FO-definable CSP's are precisely those
with  {\em finite duality}, i.e. those target structures $\BA$ for
which there exists a finite set $\mathcal F$ of structures such
that $\BB$ admits no homomorphism to $\BA$ precisely if some
structure in $\mathcal F$ admits a homomorphism to $\BB$.
 This result was followed closely by the more general
result for homomorphism-closed classes by Rossman \cite{ros}.

In this paper, we give several equivalent characterisations of
FO-definable CSP's. We first give a characterisation with an
algebraic flavour: core structures with an FO-definable CSP are
characterised by the existence of special near-unanimity
operations preserving their basic relations (Theorem
\ref{core_nuf}). For general structures, we prove that the problem
of determining if $\BA$-CSP is first-order definable is {\bf
NP}-complete (Theorem \ref{foisnpc}); if the structure $\BA$ is a
core, then in fact there exists a simple polynomial-time algorithm
to determine this (Theorem \ref{focoreispoly}). We shall also
describe in this case a simple algorithm that produces a solution
in polynomial-time (Theorem \ref{distohom}). Let $\BA$ be a core
structure such that $\BA$-CSP is first-order definable, and let
$\BB$ be a structure with the same universe, such that the basic
relations of $\BB$ are constraint relations ``inferred'' from
those of $\BA$, i.e. each is describable by a primitive positive
formula with atomic formulas of the form $\overline{x} \in \theta$
with $\theta$ a basic relation of $\BA$; these inferred relations
play a crucial role in the study of the complexity of CSP's (see
e.g. \cite{CohenJ06}). It is known that $\BB$-CSP is logspace
reducible to $\BA$-CSP \cite{jea}, but in general it will not be
first-order definable. As a simple application of our algebraic
characterisation of first-order definable CSP's (Corollary
\ref{1antoa} and Proposition \ref{nuf}) we describe precisely
which $\BB$-CSP are first-order definable; the others turn
out to be {\bf L}-complete, with their complement definable in
symmetric Datalog \cite{ELT07}, a fragment of linear Datalog.

To illustrate briefly the above results, we outline the algorithms
in the special case of digraphs. For two vertices $u, v$ of a
digraph $H$, we say that $v$ {\em dominates} $u$ if every
outneighbour of $u$ is also an outneighbour of $v$ and every
inneighbour of $u$ is also an inneighbour of $v$. If there exists
a sequence $H = H_0, H_1, \ldots, H_n = R$ of digraphs such that
$H_i$ is obtained from $H_{i-1}$ by removing a dominated vertex
for $i = 1, \ldots, n$, we say that $H$ {\em dismantles} to $R$.
More generally, $R$ is a {\em retract} of $H$ if there exists a
homomorphism from $H$ to $R$ whose restriction to $R$ is the
identity. The square $R^2$ of a digraph $R$ has vertex set $R^2$
where the arcs are the couples $((u_0,u_1),(v_0,v_1))$ such that
$(u_0,v_0)$ and $(u_1,v_1)$ are arcs of $R$, and its diagonal
$\Delta_{R^2}$ is the set of vertices of $R^2$ with both
coordinates equal.

The main algorithm to determine whether $H$-CSP is first-order
definable proceeds as follows: in $H^2$, remove any dominated
element outside the diagonal, if any. Repeat this procedure until
no element can be removed. If the resulting set is the diagonal,
then the problem is first-order definable. Assuming that $H$ is a
core, i.e. that it has no proper retract, then the converse also
holds.

\begin{figure}[htb]
\begin{center}
\includegraphics[scale=0.65]{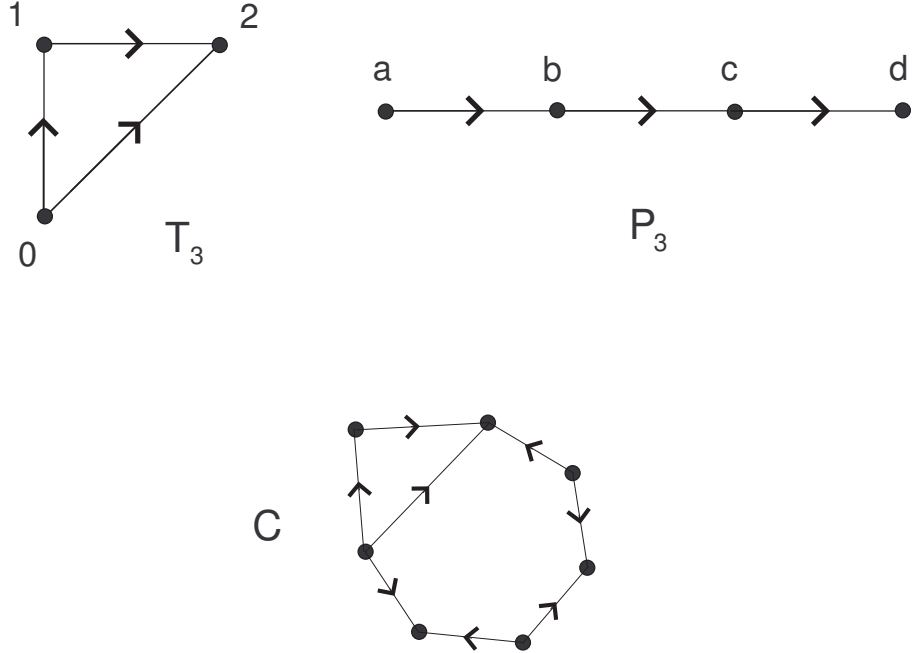}
\caption{The digraphs $T_3$, $P_3$ and $C$.}
\end{center}
\end{figure}

In the figure above, $T_3$ is the transitive tournament on three
vertices. In $T_3^2$, the two isolated vertices $(0,2), (2,0)$ are
dominated by all other vertices, the sources $(0,1), (1,0)$ are
dominated by $(0,0)$ and the sinks $(1,2), (2,1)$ are dominated by
$(2,2)$. Hence $T_3^2$ dismantles to $\Delta_{T_3^2}$, which shows
that $T_3$-CSP is first-order definable. In fact it is well known
that a directed graph $G$ admits a homomorphism to $T_3$ if and
only if there is no homomorphism from the directed 3-path $P_3$ to
$G$, and this condition is described by the first-order sentence
$\neg \, \exists \, a \, \exists \, b \, \exists \, c \, \exists
\, d \, ( A(a,b) \wedge A(b,c) \wedge A(c,d) )$, where $A(x,y)$
denotes the existence of an arc from $x$ to $y$. $P_3$-CSP is not
first-order definable; indeed the path $P_3$ is a core and $P_3^2$
can only be dismantled down to $P_3^2 \setminus \{ (a,d), (d,a)
\}$. The square of $C$ cannot be dismantled to its diagonal, but
$C$ admits $T_3$ as a retract, whence $C$-CSP is first-order
definable. Also, it is easy to check that $C \times T_3$ dismantles
to the ``graph'' $\{ (x, \phi(x)) : x \in C \}$ of a homomorphism
$\phi: C \rightarrow T_3$. In Section \ref{psfocsp},
we will see that such dismantlings of products can always be used
to produce solutions of first-order definable constraint satisfaction
problems.

\section{Preliminaries} \label{sectionprelim} For basic notation
and terminology with follow mainly \cite{dalkolvar} and
\cite{nestar}. A {\em vocabulary} is a finite set $\sigma =
\{R_1,\dots,R_m\}$ of {\em relation symbols}, each with an {\em
arity} $r_i$ assigned to it. A $\sigma$-structure is a relational
structure $\BA = \langle A;R_1(\BA),\dots,R_m(\BA)\rangle$ where
$A$ is a non-empty set called the {\em universe} of $\BA$, and
$R_i(\BA)$ is an $r_i$-ary relation on $A$ for each $i$.
We will use the same capital letter in blackboard bold and slanted
typeface to denote a structure and its universe respectively.
The elements of $R_i(\BA)$, $1\leq i \leq m$ will be called {\em
hyperedges} of $\BA$. For $\sigma$-structures $\BA$ and $\BB$, a
{\em homomorphism} from $\BA$ to $\BB$ is a map $f:A \rightarrow
B$ such that $f(R_i(\BA)) \subseteq R_i(\BB)$ for all
$1=1,\dots,m$, where for any relation $R \in \sigma$ of arity $r$
we have
$$f(R) = \{(f(x_1),\dots,f(x_r)):(x_1,\dots,x_r)
\in R\}.$$ A $\sigma$-structure $\BB$ is a {\em substructure} of a
$\sigma$-structure $\BA$ if $B \subseteq A$ and the identity map
on $B$ is a homomorphism from $\BB$ to $\BA$. For a subset $B$ of
$A$, the {\em substructure $\BB$ of $\BA$ induced by $B$} is the
$\sigma$-structure with universe $B$ with relations
$R_i(\BB)=R_i(\BA) \cap B^{r_i}$ for every $i$.
 A substructure $\BB$ of
$\BA$ is called a {\em retract} of $\BA$ if there exists a
homomorphism $\rho$ from $\BA$ to $\BB$ whose restriction to $B$
is the identity; the map $\rho$ is then called a {\em retraction}.
A structure $\BA$ is called a {\em core} if it has no retract
other than itself. It is well known (see \cite{nestar}) that every
(finite) $\sigma$-structure has a core which is unique up to
isomorphism.

Let $\BA$ be a $\sigma$-structure. We define the {\em incidence
multigraph} $\ig(\BA)$ of $\BA$ as the bipartite multigraph with
parts $A$ and $ \Block(\BA)$ which consists of all pairs $(R,r)$
such that $R \in \sigma$ and  $r \in R(A)$,  and with edges
$e_{a,i,B}$ joining $a \in A$ to $B = (R,(x_1, \ldots, x_{r})) \in
\Block(\BA)$ when $x_i = a$. This allows us to import some basic
concepts from graph theory: the {\em distance} $d_{\BA}(a,b)$
between two elements $a$ and $b$ of $A$ is defined as half their
distance in $\ig(\BA)$, the {\em diameter} of $\BA$ is defined as
half the diameter of $\ig(\BA)$, and the {\em girth} of $\BA$ is
defined as half the shortest length of a cycle in $\ig(\BA)$. In
particular, $\BA$ has girth $1$ if and only if $\ig(\BA)$ has
parallel edges, and infinite girth if and only if $\ig(\BA)$ is
acyclic. Notice in particular that tuples with repeated entries
(such as $(a,a,b)$) create parallel edges and hence cycles; this
property is not captured in the Gaifman graph. We'll require a
finer notion of tree  below and this explains why we choose this
variant of a (multi)graph associated to a relational structure
rather than the Gaifman graph.

Although this presentation of the girth differs from that given in
\cite{fedvar2}, the concept is the same and we can use the
following Erd\H{o}s-type result.
\begin{lemma}[\cite{fedvar2} Theorem 5] \label{largegirth}
Let $\BA$ and $\BB$ be $\sigma$-structures such that there exist
no homomorphism from $\BA$ to $\BB$. Then for any positive integer
$n$ there exists a $\sigma$-structure $\BA_n$ of girth greater
than $n$ such that there exists a homomorphism from $\BA_n$ to
$\BA$ but no homomorphism from $\BA_n$ to $\BB$.
\end{lemma}
\noindent Note that a $\sigma$-structure of large girth must have
large diameter unless it is acyclic.

A {\em loop} in a $\sigma$-structure $\BA$ is an element $a \in A$
such that $(a,\dots,a) \in R_i(\BA)$ for any $i$; equivalently, $a
\in A$ is a loop if and only if for every $\sigma$-structure $\BB$
the constant map $\BB \rightarrow \BA$ with value $a$ is a
homomorphism. In particular, the image of a loop under a
homomorphism is itself a loop. For an integer $n$ the {\em
$n$-link} of type $\sigma = \{R_1,\dots,R_m\}$ is the
$\sigma$-structure
$$\BL_n = \langle
\{0, 1, \ldots, n\};R_1(\BL_n),\dots,R_m(\BL_n)\rangle,$$ such
that $R_i(\BL_n) = \cup_{j = 1}^n \{j-1, j\}^{r_i}$ for $i = 1,
\ldots, m$ (where $r_i$ is the arity of the relation $R_i$). Note
that every $i \in \{0, 1, \ldots, n\}$ is a loop in $\BL_n$. A
{\em link} in an arbitrary $\sigma$-structure is a homomorphic
image of $\BL_n$ for some $n$. The term ``path'' is more common
than ``link'',  but we chose the latter to make it clear that
these are not trees in the sense defined below.

\subsection{Trees} \label{subsectiontrees}

A $\sigma$-structure $\BT$ is called a {\em $\sigma$-tree} (or
{\em tree} for short) if $\ig(\BT)$ is a tree, i.e. it is acyclic
and connected. We require the following technical results:

\begin{lemma} \label{treedecomposition}
For every $\sigma$-tree $\BT$ with $n$ hyperedges, there is a
sequence $\BT = \BT_n, \BT_{n-1}, \ldots, \BT_1$ of subtrees of
$\BT$ with the following properties: for each $j = 1, \ldots, n-1$
\begin{enumerate} \item
 $\BT_{j}$ has $j$ hyperedges; \item $\BT_{j}$ is a subtree of
 $\BT_{j+1}$; \item if $(x_1,\dots,x_r)$ is the  hyperedge of $\BT_{j+1}$ which
 does not belong to $\BT_{j}$ then there exists a unique
 index $i$ such that $x_i$
 is in the universe of $\BT_{j}$.
\end{enumerate}

\end{lemma}

\begin{proof}
Let $u_0, u_1, \ldots, u_k$ be a path of maximal length in
$\ig(\BT)$, where $\BT = (T,R_1, \ldots, R_m)$ has more than one
hyperedge. If $u_0 = (R,(x_1)) \in \Block(\BT)$, ($R$ has to be a
$1$-ary relation of $\sigma$), we obtain a new tree $\BT'$ from
$\BT$ by removing $x_1$ from $R$. If $u_0 \in T$, then $u_1 = (R,
(x_1, \ldots, x_r)) \in \Block(\BT)$, and we obtain a new tree
$\BT'$ by removing $ (x_1, \ldots, x_r)$ from $R$ and $\{x_1,
\ldots, x_m\} \setminus \{u_2\}$ from $T$. Repeating this
proceedure, we eventually obtain the desired decomposition.
\end{proof}

\begin{lemma} \label{treesofagivendiameter}
Let $\sigma = \{R_1, \ldots, R_m\}$ be a vocabulary. Then for any
integer $n$ the number of core $\sigma$-trees of diameter at most
$n$ is finite.
\end{lemma}

\begin{proof} We will show that the number $t_n$ of core {\em rooted}
trees in which the distance to the root is at most $n$ is finite.
Let $m$ be the number of relations in $\sigma$ and let $r$ be the
maximum arity of a relation in $\sigma$. We have $t_0 \leq 2^m$,
with equality only if $r = 1$. Now suppose that $t_{n-1}$ is
finite. For a rooted tree $\BT$ in which the distance to the root
$u$ is at most $n$, we can encode each hyperedge $(x_1, \ldots,
x_{r'})$ to which $u$ belongs by the name of the relation
$R_i(\BT)$ containing it (there are at most $m$ choices), the
index $i$ such that $u = x_i$ (there are at most $r$ choices) and
the trees rooted at $x_j, j \neq i$ branching away from $u$ (there
are at most $t_{n-1}^{r-1}$ choices). If $\BT$ is a core, no two
hyperedges can have the same label and $\BT$ is determined by its
set of labels of hyperedges containing $u$. Therefore $t_n \leq
2^{m\cdot r \cdot t_{n-1}^{r-1}}$. \end{proof}

\subsection{Complete sets of obstructions} \label{subsectioncomplete}

The $\sigma$-structure $\BB$ is an {\em obstruction} for the
$\sigma$-structure $\BA$ if there is no homomorphism from $\BB$ to
$\BA$. A family ${\mathcal F}$ of obstructions for $\BA$ is called
a {\em complete set of obstructions} if for every
$\sigma$-structure $\BB$ that does not admit a homomorphism to
$\BA$ there exists some $\BC \in {\mathcal F}$ which admits a
homomorphism to $\BB$. The structure $\BA$ is said to have {\em
tree duality} if it admits a complete set of obstructions
consisting of trees, and {\em finite duality} if it admits a
finite complete set of obstructions. According to \cite{nestar},
for every finite family ${\mathcal F}$ of $\sigma$-trees, there
exists a $\sigma$-structure $\BA_{\mathcal F}$ which admits
${\mathcal F}$ as a complete set of obstructions; and conversely
every $\sigma$-structure $\BA$ with finite duality admits a
finite complete set of obstructions consisting of trees. Thus the
structures with finite duality form a subclass of the structures
with tree duality, and there is one such core structure for every
finite set of tree obstructions.

An obstruction $\BB$ for $\BA$ is called {\em critical} if every
proper substructure of $\BB$ admits a homomorphism to $\BA$. It is
clear that a critical obstruction is a core, and that every
obstruction contains, as a substructure, a critical obstruction.

\begin{lemma} \label{fd=bd}
A $\sigma$-structure $\BA$ has finite duality if and only if
there is an upper bound on the diameter of its critical
obstructions.
\end{lemma}
\begin{proof}
Clearly, if $\BA$ has finite duality, then the maximum diameter
of an obstruction in a finite complete set of obstructions for
$\BA$ is an upper bound on the diameter of all critical
obstructions for $\BA$. Conversely, suppose that the critical
obsructions for $\BA$ have diameter at most $m$. Let ${\mathcal
F}$ be the set of core $\sigma$-trees of diameter at most $m$
which do not admit a homomorphism to $\BA$. By
Lemma~\ref{treesofagivendiameter}, ${\mathcal F}$ is finite. By
Lemma~\ref{largegirth}, for any $\sigma$-structure $\BB$ which
does not admit a homomorphism to $\BA$, there exists a structure
$\BC$ of girth at least $2m + 2$ which admits a homomorphism to
$\BB$ but not to $\BA$. A critical obstruction for $\BA$ contained
in $\BC$ cannot contain a cycle hence it must be a tree $\BT$ of
diameter at most $m$. Therefore $\BT \in {\mathcal F}$; this shows
that ${\mathcal F}$ is a finite complete set of obstructions for
$\BA$.
\end{proof}

For a $\sigma$-structure $\BA$, the problem $\BA$-CSP
 consists of determining whether an input structure
$\BB$ admits a homomorphism to $\BA$. It is said to be {\em
first-order definable} if there exists a first-order sentence
$\Phi$ (in the language of $\sigma$)  which is true on $\BB$ if
and only if $\BB$ admits a homomorphism to $\BA$. By a result of
Atserias \cite{ats}, $\BA$-CSP is first-order definable if and
only if $\BA$ has finite duality, hence we have the following
equivalences:
\begin{theorem} \label{foequivalences}
Let $\BA$ be a $\sigma$-structure. Then the following are
equivalent.
\begin{itemize}
\item[1.] $\BA$-CSP is first-order definable; \item[2.] $\BA$ has
finite duality; \item[3.] $\BA$ has a finite complete set of
obstructions consisting of trees; \item[4.] The critical
obstructions of $\BA$ have bounded diameter.
\end{itemize}
\end{theorem}
We are mostly interested in the ``meta-problem'' of deciding
whether an input structure $\BA$ has a first-order definable CSP.
The equivalences of Theorem~\ref{foequivalences} are not a usable
decision procedure, but they will be used in the next two sections
to find such a procedure. As a benchmark we state here Feder and
Vardi's decision procedure for tree duality. Given a structure
$\BA = \langle A;R_1(\BA),\dots,R_m(\BA)\rangle$, we define the
structure ${\mathcal U}(\BA)  = \langle U;R_1({\mathcal
U}(\BA)),\dots,R_m({\mathcal U}(\BA))\rangle$, where $U$ is the
set of all nonempty subsets of $A$, and for $i = 1, \ldots, m$,
$R_i({\mathcal U}(\BA))$ is the set of all $r_i$-tuples $(X_1,
\ldots, X_{r_i})$ such that for all $j \in \{1, \ldots, r_i\}$ and
$x_j \in X_j$ there exist $x_k \in X_k, k \in \{1, \ldots,
r_i\}\setminus \{j\}$ such that $(x_1, \ldots, x_{r_i}) \in
R_i(\BA)$.
\begin{theorem}[\cite{fedvar2} Theorem 21] \label{tddecidable}
A $\sigma$-structure $\BA$ has tree duality if and only if there
exists a homomorphism from ${\mathcal U}(\BA)$ to $\BA$.
\end{theorem}
This proves that determining whether a given structure $\BA$ has
tree duality is decidable, since a search for a homomorphism from
${\mathcal U}(\BA)$ to $\BA$ can be done in finite time. In
section \ref{sectionconstructions}, we provide similar
``construction-and-homomorphism'' characterisations of first-order
definable constraint satisfaction problems. However, we must first
clear up a technical point concerning tree duality: indeed, Feder
and Vardi's definition of a tree given in \cite{fedvar2} (at the
bottom of page 79) is slightly more general than the one given
here, as it allows parallel edges in the incidence multigraph of a
tree. Nonetheless the corresponding concepts of ``tree duality''
turn out to be equivalent, as we show in the next section.

\section{Tree Duality}

In this section we prove that the notion of tree duality is the
same whether we use the notion of tree as defined here or as
defined in \cite{fedvar2}.

\begin{lem} \label{treetoua}
Let $\BT$ be a tree. Then $\BT$ admits a homomorphism to $\BA$ if
and only if $\BT$ admits a homomorphism to ${\mathcal U}(\BA)$.
\end{lem}

\proof First note that the singletons induce a copy of $\BA$ in
${\mathcal U}(\BA)$. Thus, if a tree $\BT$ admits a homomorphism
to $\BA$, it also admits a homomorphism to ${\mathcal U}(\BA)$.

Conversely, suppose that $\phi: \BT \rightarrow {\mathcal U}(\BA)$
is a homomorphism. Let $\BT = \BT_n, \BT_{n-1}, \ldots,  \BT_1$ be
a decomposition of $\BT$ as described in Lemma
\ref{treedecomposition}. We can define a sequence of homomorphisms
$\psi_i : \BT_i \rightarrow \BA$ as follows: $\BT_1$ has a single
hyperedge $(x_1,\dots,x_r) \in R(\BA)$ for some $R \in \sigma$.
Since $\phi$ is a homomorphism, by definition of ${\mathcal
U}(\BA)$ there exist $\psi_1(x_j) \in \phi(x_j)$, $j = 1, \ldots,
r$, such that $(\psi_1(x_1),\dots,\psi_1(x_r)) \in R({\mathcal
U}(\BA))$. This defines a homomorphism $\psi_1 : \BT_1 \rightarrow
\BA$.

Now suppose that $i < n$ and $\psi_i : \BT_i \rightarrow \BA$ is
already defined. $\BT_{i+1}$ is obtained from $\BT_i$ by adding
one hyperedge $(x_1, \ldots, x_r) \in R(\BT)$ for some $R \in
\sigma$, where for exactly one index $j_0$, $x_{j_0}$ belongs to
the universe of $\BT_i$. Since $\phi$ is a homomorphism, by
definition of ${\mathcal U}(\BA)$ there exist $\psi_{i+1}'(x_j)
\in \phi(x_j)$, $j = 1, \ldots, r$, such that
$(\psi_{i+1}'(x_1),\dots,\psi_{i+1}'(x_r)) \in R({\mathcal
U}(\BA))$ and $\psi_{i+1}'(x_{j_0}) = \psi_i(x_{j_0})$. Then
$\psi_i \cup \psi_{i+1}' = \psi_{i+1}$ is a well defined
homomorphism from $\BT_{i+1}$ to $\BA$. Continuing in this way, we
eventually define a homomorphism $\psi = \psi_n$ from $\BT$ to
$\BA$. \qed

The {\em $\BA$ hyperedge consistency check} is the following
polynomial-time algorithm. At the start every element $b \in B$
is assigned a list consisting of all the
``plausible'' images of $b$ under such a homomorphism; initially
this list is $A$.
Then, the elements of $B$ are cyclically
inspected to check whether their lists are still consistent with
the local information: For every element $b$, an element $a$ in
the current list of $b$ is removed if there exists a relation $R
\in \sigma$ and $(b_1, \ldots, b_r) \in R(\BB)$ with $b_j = b$ for
some $j$ such that there exists no $(a_1, \ldots, a_r) \in R(\BA)$
with $a_j = a$ and $a_i$ in the list of $b_i$ for $i \neq j$. The
process continues until the lists stabilise. If at some point the
list of an element becomes empty, the $\BA$ hyperedge consistency
check is said to fail on $\BB$, and otherwise it is said to
succeed on $\BB$. The following result gives an interpretation of
these possible outcomes.
\begin{lem} \label{hcc}
The $\BA$ hyperedge consistency check succeeds on $\BB$ if and
only if there exists a homomorphism from $\BB$ to ${\mathcal
U}(\BA)$, and it fails on $\BB$ if and only if there exists a tree
$\BT$ which admits a homomorphism to $\BB$ but no homomorphism to
$\BA$.
\end{lem}

\proof If the $\BA$ hyperedge consistency check succeeds on $\BB$,
then by definition of ${\mathcal U}(\BA)$ the map
$\phi: \BB \rightarrow {\mathcal U}(\BA)$ assigning to
every $b \in B$ its final list
$\phi(b)$ is a homomorphism. Conversely, if $\phi: \BB \rightarrow
{\mathcal U}(\BA)$ is a homomorphism, then again by definition of
${\mathcal U}(\BA)$, the $\BA$ hyperedge consistency check will
never eliminate an element $c \in \phi(b)$ from the list of any
element $b$ of $\BB$, hence it will succeed.

If there exists a tree $\BT$ which admits a homomorphism to $\BB$
but not to $\BA$, then by Lemma \ref{treetoua}, $\BT$ does not
admit a homomorphism to ${\mathcal U}(\BA)$, thus $\BB$ does not
admit a homomorphism to ${\mathcal U}(\BA)$. By the previous
paragraph this implies that the
 $\BA$ hyperedge consistency check must fail on $\BB$.
Conversely, suppose that the
 $\BA$ hyperedge consistency check fails on $\BB$.
We construct a tree $\BT$ as follows while running the hyperedge
consistency check.

When deleting an element $a$ from the list of an element $b$ of
$\BB$, we define a rooted tree $\BT_{a,b}$ with root $r_{a,b}$
with the following properties.
\begin{itemize}
\item[(i)] There is a homomorphism from $\BT_{a,b}$ to $\BB$
mapping $r_{a,b}$ to $b$, \item[(ii)] there is no homomorphism
from $\BT_{a,b}$ to $\BA$
 mapping $r_{a,b}$ to $a$.
\end{itemize}
Indeed, $a$ is deleted from the list of $b$ because we found a
relation $R \in \sigma$ and $(b_1, \ldots, b_r) \in R(\BB)$ with
$b_j = b$ for some $j$ such that there exists no $(a_1, \ldots,
a_r) \in R(\BA)$ with $a_j = a$ and $a_i$ in the list of $b_i$ for
$i \neq j$. We then put elements $c_1, \ldots, c_r$ in the
universe of $\BT_{a,b}$ and $(c_1, \ldots, c_r)$ in $R(\BT_{a,b})$,
and select the root $r_{a,b} = c_j$. If $R(\BA)$ does not contain
any $r$-tuple $(a_1,\ldots,a_r)$ such that $a_j = a$, then we are
done, since $\BT_{a,b}$ clearly satisfies properties (i) and (ii).
Otherwise, for every $(a_1,\ldots,a_r)$ such that $a_j = a$, there
exists at least one index $i$ such that $a_i$ is already removed
from the list of $b_i$ whence the tree $\BT_{a_i,b_i}$ with
properties (i) and (ii) is already defined; we then add a copy of
$\BT_{a_i,b_i}$ to $\BT_{a,b}$, by identifying $r_{a_i,b_i}$ to
$c_i$. Thus we get a tree $\BT_{a,b}$ such that the map $\phi:
\{c_1,\ldots, c_r\} \rightarrow \{b_1,\ldots,b_r\}$ defined by
$\phi(c_i) = b_i$ extends to a homomorphism from $\BT_{a,b}$ to
$\BB$. However, any homomorphism from $\BT_{a,b}$ to $\BA$
mapping $r_{a,b}$ to $a$ would also
map the root of some previously defined $\BT_{a_i,b_i}$ to $a_i$,
which is impossible. Thus $\BT_{a,b}$ satisfies properties (i) and
(ii).

When the list of some element $b$ of $\BB$ becomes empty, we can
construct a tree $\BT$ by identifying the roots of all the trees
$\BT_{a,b}$ to a new element $r$. We then find a homomorphism from
$\BT$ to $\BB$ by mapping $r$ to $b$ and extending independently
on each $\BT_{a,b}$. However a homomorphism from $\BT$ to $\BA$
would need to map $r$ to some element $a$, hence induce a
homomorphism from $\BT_{a,b}$ to $\BA$ mapping its root to $a$,
which is impossible. Therefore, when the hyperedge consistency check
fails, there exists a tree $\BT$ which admits a homomorphism to
$\BB$ but not to $\BA$. \qed

In terms of dualities, these results can be summarized as follows.

\begin{thm}
For a $\sigma$-structure $\BA$, the following properties are
equivalent:
\begin{itemize}
\item[(i)] $\BA$ has tree duality, \item[(ii)] ${\mathcal U}(\BA)$
admits a homomorphism to $\BA$, \item[(iii)] The $\BA$ hyperedge
consistency check decides the $\BA$-CSP problem.
\end{itemize}
\end{thm}

\proof \hfill
\begin{itemize}
\item[(i)] $\Rightarrow$ (ii)\ Suppose that $\BA$ has tree duality.
By Lemma \ref{treetoua}, there does not exist a tree that admits a
homomorphism to ${\mathcal U}(\BA)$ but not to $\BA$, hence there
exists a homomorphism from ${\mathcal U}(\BA)$ to $\BA$.
\item[(ii)] $\Rightarrow$ (iii)\ By Lemma \ref{hcc} the $\BA$
hyperedge consistency check decides the ${\mathcal U}(\BA)$-CSP
problem. If there exists a homomorphism from  ${\mathcal U}(\BA)$
to $\BA$, then since there also exists a homomorphism from $\BA$
to ${\mathcal U}(\BA)$ the $\BA$-CSP problem is equivalent to the
${\mathcal U}(\BA)$-CSP problem. \item[(iii)] $\Rightarrow$ (i)\ By
Lemma \ref{hcc}, when the $\BA$ hyperedge consistency check fails
on a structure $\BB$, there exists a tree $\BT$ which admits a
homomorphism to $\BB$ but none to $\BA$. Thus if the $\BA$
hyperedge consistency check decides the $\BA$-CSP problem, then
$\BA$ has tree duality.
\end{itemize}
\qed

Property (ii) shows that these three properties are decidable,
since a greedy search for a homomorphism from ${\mathcal U}(\BA)$
to $\BA$ can be done in finite time. Property (iii) gives a
polynomial algorithm for the corresponding CSP's, and property (i)
shows that these problems contain the class of CSP's with finite
duality, that is, the first-order decidable CSP's.

\section{Constructions} \label{sectionconstructions}
\subsection{Quotients} \label{subsectionquotients}
Let $\BA = \langle A;R_1(\BA),\dots,R_m(\BA)\rangle$ be a
$\sigma$-structure and $\sim$ an equivalence relation on $A$. For
$a \in A$ we denote $a/\!\!\sim$ the $\sim$-equivalence class
containing $a$. The {\em quotient} $\BA/\!\!\sim$ of $\BA$ under
$\sim$ is the $\sigma$-structure whose universe is the set of
$\sim$-equivalence classes, where for $i = 1, \ldots, m$ we have
$(C_1, \ldots, C_{r_i}) \in R_i(\BA/\!\!\sim)$ if and only if
there exist $a_j \in C_j$, $j = 1, \ldots, r_i$ such that $(a_1,
\ldots, a_{r_i}) \in R_i(\BA)$. Note that the quotient map $q: \BA
\rightarrow \BA/\!\!\sim$ where $q(a) = a/\!\!\sim$ is a
homomorphism; in fact for every homomorphism $\phi: \BA
\rightarrow \BB$, there is a natural equivalence $\sim$ (the
``kernel'' of $\phi$) on $A$ and an injective homomorphism $\psi:
\BA/\!\!\sim \rightarrow \BB$ such that $\phi = \psi \circ q$.

Here we give a first application of quotients to reveal an
important structural property of cores with tree duality. A
$\sigma$-structure $\BA$ is called {\em rigid} if the identity is
the only homomorphism from $\BA$ to itself.
\begin{lemma} \label{coreisrigid} Let $\BA$ be a core with tree duality.
Then $\BA$ is rigid. \end{lemma}
\begin{proof} Suppose that $\tau: \BA \rightarrow \BA$
is a homomorphism. Since $\BA$ is a core, $\tau$ is an
automorphism of $\BA$ hence we can define an equivalence relation
$\sim$ on $A$ by putting $a \sim b$ if there exists an integer $p$
such that $\tau^p(a) = b$. We will show that every tree which
admits a homomorphism to $\BA/\!\!\sim$ also admits a homomorphism
to $\BA$.

Let $\BT$ be a tree which admits a homomorphism $\psi: \BT
\rightarrow \BA/\!\!\sim$. Let $\BT = \BT_n, \BT_{n-1}, \ldots,
\BT_1$ be the sequence of Lemma~\ref{treedecomposition}. For $k =
1, \ldots, n$, the restriction of $\psi$ to the universe of
$\BT_k$ is a homomorphism $\psi_k: \BT_k \rightarrow
\BA/\!\!\sim$; we recursively define a sequence $\phi_k: \BT_k
\rightarrow \BA$ of homomorphisms such that $\psi_k = q \circ
\phi_k$, where $q$ is the quotient map from $\BA$ to
$\BA/\!\!\sim$. First, $\BT_1$ has just one hyperedge $(x_1,
\ldots, x_{r_i}) \in R_i(\BT)$ for some $i$, and $(\psi(x_1),
\ldots, \psi(x_{r_i})) \in R_i(\BA/\!\!\sim)$. By definition of
quotients this means that there exist $y_j \in \psi(x_j), j = 1,
\ldots, r_i$ such that $(y_1, \ldots, y_{r_i}) \in R_i(\BA)$, thus
we can define $\phi_1: \BT_1 \rightarrow \BA$ by $\phi_1(x_j) =
y_j$. Now suppose that $\phi_{k-1}: \BT_{k-1} \rightarrow
\BA/\!\!\sim$ is already defined. $\BT_k$ is obtained from
$\BT_{k-1}$ by adding an hyperedge $(x_1, \ldots, x_{r_i}) \in
R_i(\BT)$ which has only one coordinate $x_{\ell}$ in the universe
of $\BT_{k-1}$. Again we have $(\psi(x_1), \ldots, \psi(x_{r_i}))
\in R_i(\BA/\!\!\sim)$ and there exist $y_j \in \psi(x_j), j = 1,
\ldots, r_i$ such that $(y_1, \ldots, y_{r_i}) \in R_i(\BA)$. Put
$a = \phi_{k-1}(x_{\ell}) \in a/\!\!\sim = \psi(x_{\ell})$. Then
$y_{\ell} \sim a$ hence by the definition of $\sim$ there exists a
power $p$ such that $\tau^p(y_{\ell}) = a$. Since $\tau$ is a
homomorphism, we then have $(\tau^p(y_1), \ldots, \tau^p(y_{r_i}))
\in R_i(\BA)$, and we can extend the definition of $\phi_{k-1}$ to
that of $\phi_k: \BT_k \rightarrow \BA$ by putting $\phi_k(z) =
\phi_{k-1}(z)$ if $z$ is in the universe of $\BT_{k-1}$, and
$\phi_k(x_j) = \tau^p(y_j), j = 1, \ldots, r_i$. Indeed $\phi_k$
is well defined since both definitions coincide on $x_{\ell}$, it
is a homomorphism since it preserves $(x_1, \ldots, x_{r_i}) \in
R_i(\BT)$ in addition to all the hyperedges preserved by
$\phi_{k-1}$, and $\phi_k(z) \in \psi_k(z)$ for all $z$ in the
universe of $\BT_k$ whence $\psi_k(z) = q \circ \phi_k(z)$. In
this way we eventually define a homomorphism $\phi = \phi_n$ from
$\BT = \BT_n$ to $\BA$.

Hence every tree which admits a homomorphism to $\BA/\!\!\sim$
also admits a homomorphism to $\BA$. Since $\BA$ has  tree duality
this implies that $\BA/\!\!\sim$ admits a homomorphism to $\BA$.
Since $\BA$ is a core which admits a homomorphism to
$\BA/\!\!\sim$, this implies that $\sim$ cannot identify vertices,
whence $\tau$ is the identity.
 \end{proof}

\subsection{Products and powers} \label{sectionproducts}
Given two $\sigma$-structures $\BA = \langle
A;R_1(\BA),\dots,R_m(\BA)\rangle$ and $\BB = \langle
B;R_1(\BB),\dots,R_m(\BB)\rangle$ their {\em product} is the
$\sigma$-structure
$$\BA \times \BB = \langle A\times B;
R_1(\BA \times \BB),\dots,R_m(\BA \times \BB)\rangle,$$ where for
$i = 1, \ldots, m$, $R_i(\BA \times \BB)$ consists of all tuples
$((a_1,b_1), \ldots, (a_{r_i},b_{r_i}))$ such that $(a_1, \ldots,
a_{r_i}) \in R_i(\BA)$ and $(b_1, \ldots, b_{r_i}) \in R_i(\BB) $.
Both projections $\pi_1: \BA \times \BB \rightarrow \BA$ and
$\pi_2: \BA \times \BB \rightarrow \BB$ are homomorphism and in
general for any $\sigma$-structure $\BC$ and any pair $\phi_1: \BC
\rightarrow \BA$, $\phi_2: \BC \rightarrow \BB$ of homomorphisms
there is a unique homomorphism $\phi: \BC \rightarrow \BA \times
\BB$ such that $\phi_1 = \pi_1 \circ \phi$ and $\phi_2 = \pi_2
\circ \phi$. The product is associative; the {\em $n$-th power}
$\BA^n$ of $\BA$ is the product of $n$ copies of $\BA$. For any $n
\geq 1$ an {\em $n$-ary operation on} $\BA$ is a homomorphism from
$\BA^n$ to $\BA$.

The {\em one-tolerant $n$-th power} $^1\BA^n$ of $\BA$ is the
$\sigma$-structure $\langle A^n; R_1(^1\BA^n), \ldots,
R_m(^1\BA^n) \rangle$ where for $i = 1, \ldots, m$, $R_i(^1\BA^n)$
consists of tuples $((a_{1,1}, \ldots, a_{1,n}), \ldots,
(a_{r_i,1}, \ldots, a_{r_i,n}))$ such that $\vert \{ k : (a_{1,k},
\ldots, a_{r_i,k}) \in R_i(\BA) \} \vert \geq n-1$. In other
words, $^1\BA^n$ is obtained from $\BA^n$ by adding to
$R_i(\BA^n)$ all hyperedges that are mapped to $R_i(\BA)$ by at
least $n-1$ of the projections. In particular, the projections are
not homomorphisms from $^1\BA^n$ to $\BA$ hence $^1\BA^n$ does not
necessarily admit a homomorphism to $\BA$. However notice that
removal of a coordinate is a homomorphism from $^1\BA^{n+1}$ to
$^1\BA^n$.

\begin{lemma} \label{homtolerantpower}
There exists a homomorphism from $^1\BA^{n+1}$ to $\BA$ if and
only if the critical obstructions of $\BA$ have at most $n$
hyperedges.
\end{lemma}
\begin{proof}
Let $\BC$ be a critical obstruction of $\BA$ with $m$ distinct
hyperedges $e_1, \ldots, e_{m}$, $m > n$. Then for $k = 1, \ldots,
m$, the $\sigma$-structure $\BC_j$ obtained from $\BC$ by removing
$e_k$ (without changing the universe) admits a homomorphism
$\phi_k$ to $\BA$. By definition of $^1\BA^{m}$, the map $\phi =
(\phi_1, \ldots, \phi_{m})$ is a homomorphism from $\BC$ to
$^1\BA^{m}$. Therefore there is no homomorphism from $^1\BA^{m}$
to $\BA$, and in particular none from $^1\BA^{n+1}$ to $\BA$.

Conversely, suppose that there is no homomorphism from
$^1\BA^{n+1}$ to $\BA$. Then there exists a critical obstruction
$\BC$ of $\BA$ which admits a homomorphism $\phi$ to
$^1\BA^{n+1}$. For every coordinate $k = 1, \ldots, n+1$, there
exists an hyperedge $e_k$ of $\BC$ which is not respected by
$\pi_k \circ \phi$, since $\pi_k \circ \phi$ is not a homomorphism
from $\BC$ to $\BA$. By the definition of $^1\BA^{n+1}$, $e_k$ is
respected by $\pi_j \circ \phi$ for every $j \neq k$, whence $e_j
\neq e_k$ for $j \neq k$. Therefore $\BC$ has at least $n+1$
hyperedges.
\end{proof}
\begin{corollary} \label{1antoa}
A $\sigma$-structure $\BA$ has finite duality if and only if
there exists a positive integer $n$ such that $^1\BA^n$ admits a
homomorphism to $\BA$.
\end{corollary}
Note that the homomorphisms from $1$-tolerant powers of $\BA$ to
$\BA$ are operations on $\BA$. For $n \geq 3$,
an operation $\phi: \BA^n
\rightarrow \BA$ is called a {\em near unanimity operation} if it
satisfies the identities
$$\begin{array}{c}
  \phi(y,x,x,\ldots,x) = \phi(x,y,x,\ldots,x)
 = \cdots = \phi(x,x,x,\ldots,y) = x.
\end{array}
 $$
\begin{lemma} \label{nuf} Let $\BA$ be a core with finite duality.
Then every homomorphism from a 1-tolerant power of $\BA$ to $\BA$
is a near unanimity operation.
\end{lemma}
\begin{proof}
Let $\phi: {^1\BA^{n}} \rightarrow \BA$ be a homomorphism. For
every $y \in A$ and $k \in \{1, \ldots, n\}$, consider the
homomorphism $\psi_{y,k} : \BA \rightarrow {^1\BA^{n}}$ defined by
$\psi_{y,k}(x) = (x_1, \ldots, x_n)$ where $x_j = y$ if $j = k$
and $x_j = x$ otherwise. By Lemma~\ref{coreisrigid}, $\BA$ is
rigid whence the map $\phi \circ \psi_{y,k} : \BA \rightarrow \BA$
is the identity. Thus for every $x, y \in A$ and $k \in \{1,
\ldots, n\}$ we have $\phi(\psi_{y,k}(x)) = x$, and this is
precisely the definition of a near unanimity operation.
\end{proof}

\label{section_reference} We say that a structure $\BA$ {\em
admits} an operation $f:A^n\rightarrow A$, or equivalently that
$f$ {\em preserves} the basic relations of $\BA$ if $f$ is a
homomorphism from $\BA^n$ to $\BA$.

\begin{corollary} \label{core_nuf}
Every core relational structure with a first-order definable CSP
admits a near unanimity operation.
\end{corollary}

\subsection{Products of links and squares}
Recall from Section~\ref{sectionprelim} that the $n$-link $\BL_n$
of type $\sigma$ has universe $\{0, 1, \ldots, n\}$. For a
$\sigma$-structure $\BC$, a map $\phi$ from its universe to $\{0,
1, \ldots, n\}$ is a homomorphism from $\BC$ to $\BL_n$ if and
only if $\vert \phi(x) -\phi(y) \vert \leq 1$ whenever $x$ and $y$
are in a common hyperedge.

Given a $\sigma$-structure $\BA = \langle
A;R_1(\BA),\dots,R_m(\BA)\rangle$, note that the product $\BL_n
\times \BA^2$ has diameter at least $n$ since for any $a, a', b,
b' \in A$ the distance between $(0,a,b)$ and $(n,a',b')$ is at
least $n$. (The distance could even be infinite, that is,
$(0,a,b)$ and $(n,a',b')$
could lie in different connected components.)
Let $\sim_n$ be the equivalence relation defined on
$\BL_n \times \BA^2$ by
$$
(k,a,b) \sim_n (k',a',b') \equiv \left \{
\begin{array}{l}
\mbox{$(k,a,b) = (k',a',b')$} \\
\mbox{or $k = k' = 0$ and $a = a'$}\\
\mbox{or $k = k' = n$ and $b = b'$}.
\end{array}
\right.
$$
Note that $\BL_n \times \BA^2 /\!\!\sim_n$ also has diameter at
least $n$.
\begin{lemma} \label{cutends}
The substructures $\BB_0$ and $\BB_n$ of $\BL_n \times \BA^2
/\!\!\sim_n$ induced by $B_0 = \{ (k,a,b)/\!\!\sim_n : k \neq 0\}$
and $B_n = \{ (k,a,b)/\!\!\sim_n : k \neq n\}$ respectively both
admit homomorphisms to $\BA$.
\end{lemma}
\begin{proof}
On $B_0$ we can define a map $\phi$ to $A$ by
$\phi((k,a,b)/\!\!\sim_n) = b$. We show that $\phi$ is a
homomorphism from $\BB_0$ to $\BA$. For $R_i \in \sigma$ and
$((k_1,a_1,b_1)/\!\!\sim_n, \ldots,
(k_{r_i},a_{r_i},b_{r_i})/\!\!\sim_n) \in R_i(\BB_0)$, there exist
$(k_j',a_j',b_j') \in (k_j,a_j,b_j)/\!\!\sim_n$, $j = 1, \ldots,
r_i$, with $((k_1',a_1',b_1'), \ldots,
(k_{r_i}',a_{r_i}',b_{r_i}')) \in R_i(\BL_n \times \BA^2)$. We
then have that $$(\phi((k_1,a_1,b_1)/\!\!\sim_n), \ldots,
\phi((k_{r_i},a_{r_i},b_{r_i})/\!\!\sim_n))$$ is equal to $(b_1',
\ldots, b_{r_i}')$ which  is in  $R_i(\BA)$, thus $\phi$ is a
homomorphism. Similarly, we can define a homomorphism $\psi: \BB_n
\rightarrow \BA$ by $\psi((k,a,b)/\!\!\sim_n) = a$.
\end{proof}
\begin{prop} \label{lxa2toa}
A $\sigma$-structure $\BA$ has critical obstructions of bounded
diameter if and only if there exists a positive integer $n$ such
that $\BL_n \times \BA^2 /\!\!\sim_n$ admits a homomorphism to
$\BA$.
\end{prop}
\begin{proof}
By the previous lemma, any critical obstruction of $\BA$ contained
in $\BL_n \times \BA^2 /\!\!\sim_n$ must contain an element with
first coordinate $0$ and an element with first coordinate $n$ (the
first coordinates are invariants of $\sim_n$-equivalence classes)
thus have diameter at least $n$. Hence if $n$ is larger than the
diameter of all the critical obstructions of $\BA$, then $\BL_n
\times \BA^2 /\!\!\sim_n$ admits a homomorphism to $\BA$.

Now suppose that $\BA$ has critical obstructions of arbitrarily
large diameter. We will show that for every integer $n$ there
exists an obstruction $\BC$ of $\BA$ which admits a homomorphism
to $\BL_n \times \BA^2 /\!\!\sim_n$. Let $\BC = \langle C;
R_1(\BC), \ldots, R_m(\BC) \rangle$ be an obstruction of $\BA$
with diameter at least $n+2$. Let $x$ and $y$ be elements of $C$
at distance $n+2$, and $\BC_x$, $\BC_y$ the substructures of $\BC$
induced respectively by $C\setminus \{x\}$ and $C\setminus \{y\}$.
Fix homomorphisms $\alpha: \BC_y \rightarrow \BA$ and $\beta:
\BC_x \rightarrow \BA$ and define $\kappa: C \rightarrow \{0,
\ldots, n\}$ by
$$
\kappa(z) = \left \{
\begin{array}{l} \mbox{$0$ if $z = x$,} \\
 \mbox{$d_{\BC}(x,z) - 1$ if $d_{\BC}(x,z) \leq n+1$ and $z \neq x$,} \\
  \mbox{$n$ if $d_{\BC}(x,z) \geq n+2$;}
\end{array}
\right.
$$
note that $\kappa$ is a homomorphism from $\BC$ to $\BL_n$. We fix
an element $p \in A$ and define a map $\phi$ from $C$ to the
universe of $\BL_n \times \BA^2 /\!\!\sim_n$ by
$$
\phi(z) = \left \{
\begin{array}{l}
\mbox{$(\kappa(z),\alpha(z),\beta(z))/\!\!\sim_n$ if $z \neq x, y$,} \\
\mbox{$(\kappa(z),\alpha(z),p)/\!\!\sim_n$ if $z = x$,}  \\
\mbox{$(\kappa(z),p,\beta(z))/\!\!\sim_n$ if $z = y$.}
\end{array}
\right.
$$
We will show that $\phi$ is a homomorphism from $\BC$ to $\BL_n
\times \BA^2 /\!\!\sim_n$.

Let $(z_1, \ldots, z_{r_i})$ be in $R_i(\BC)$ for some $R_i \in
\sigma$. If $z_j \not \in \{x, y\}$ for all $j \in \{1, \ldots,
r_i\}$, then $(\phi(z_1), \ldots, \phi(z_{r_i})) =
((\kappa(z_1),\alpha(z_1),\beta(z_1))/\!\!\sim_n, \ldots,
(\kappa(z_{r_i}),\alpha(z_{r_i}),\beta(z_{r_i}))/\!\!\sim_n$ which
belongs to $R_i(\BL_n \times \BA^2 /\!\!\sim_n)$ since $\kappa,
\alpha, \beta$ and the quotient map from $\BL_n \times \BA^2$ to
$\BL_n \times \BA^2 /\!\!\sim_n$ are homomorphisms. If there
exists an index $\hat{j}$ such that $z_{\hat{j}} = x$, then
$\kappa(z_j) = 0$ for $j = 1, \ldots, r_i$ whence $\phi(z_j) =
(0,\alpha(z_j),\alpha(z_j))/\!\!\sim_n$ for $j = 1, \ldots, r_i$
by definition of $\sim_n$; therefore $(\phi(z_1), \ldots,
\phi(z_{r_i}))$ is equal to
$((0,\alpha(z_1),\alpha(z_1))/\!\!\sim_n, \ldots,
(0,\alpha(z_{r_i}),\alpha(z_{r_i}))/\!\!\sim_n)$ which  is in $
R_i(\BL_n \times \BA^2 /\!\!\sim_n)$. Similarly if there exists an
index $\hat{j}$ such that $z_{\hat{j}} = y$, then $(\phi(z_1),
\ldots, \phi(z_{r_i})) = ((n,\beta(z_1),\beta(z_1))/\!\!\sim_n,
\ldots, (n,\beta(z_{r_i}),\beta(z_{r_i}))/\!\!\sim_n)$ $\in
R_i(\BL_n \times \BA^2 /\!\!\sim_n)$. Thus $\phi$ is a
homomorphism.

Since there exists a homomorphism from an obstruction of $\BA$ to
$\BL_n \times \BA^2 /\!\!\sim_n$ we  conclude that there is no
homomorphism from $\BL_n \times \BA^2 /\!\!\sim_n$ to $\BA$.
\end{proof}

 By Theorem~\ref{foequivalences},
Corollary~\ref{1antoa} and Proposition~\ref{lxa2toa} we have the
following characterisations:
\begin{theorem} \label{focharacterisations}
Let $\BA$ be a $\sigma$-structure. Then the following are
equivalent.
\begin{itemize}
\item[1.] $\BA$-CSP is first-order definable; \item[2.] For some
$n$ there exists a homomorphism from $^1\BA^n$ to $\BA$; \item[3.]
For some $n$ there exists a homomorphism from $\BL_n \times \BA^2
/\!\!\sim_n$ to $\BA$.
\end{itemize}
\end{theorem}

At first glance our situation vis-a-vis the decidability question
appears no better than before, but a closer look at the third
condition in the above theorem reveals an upper bound on $n$:
indeed, for $0 \leq k \leq n$, the restriction $\phi_k$ of a
homomorphism $\phi: \BL_n \times \BA^2 /\!\!\sim_n \rightarrow
\BA$ to $\{k\}\times \BA^2 /\!\!\sim_n$ corresponds to a
homomorphism from $\BA^2$ to $\BA$, and there are at most
$|A|^{|A|^2}$ of these. If for $k < k'$ we have $\phi_k =
\phi_{k'}$, then for $n' = n - k' + k$ we can define a
homomorphism $\phi': \BL_{n'} \times \BA^2 /\!\!\sim_{n'}
\rightarrow \BA$ by removing the useless middle part. Therefore to
determine whether $\BA$-CSP is first-order definable it suffices
to search for a homomorphism $\phi: \BL_n \times \BA^2 /\!\!\sim_n
\rightarrow \BA$ with $n \leq |A|^{|A|^2}$, and this is a finite
decision procedure.

We can refine this argument by defining a graph structure on the
set of all homomorphisms from $\BA^2$ to $\BA$, where two
homomorphisms $\psi$, $\psi'$ are called {\em adjacent} if there
exists a homomorphism $\phi: \BL_1 \times \BA^2 \rightarrow \BA$
such that $\phi_0 = \psi$ and $\phi_1 = \psi'$. A homomorphism
from $\BL_n \times \BA^2 /\!\!\sim_n$ to $\BA$ then corresponds to
a link of length $n$ between a homomorphism $\phi_0: \BA^2
\rightarrow \BA$ which factors through the first projection and a
homomorphism $\phi_n: \BA^2 \rightarrow \BA$ which factors through
the second projection. Since undirected reachability can be solved
in logarithmic space, in our exponential setting this means that
the search can be performed in polynomial space. In the next
section this idea is developed further and we prove that the
problem of determining whether $\BA$-CSP is first-order definable
is  actually in {\bf NP}.

\section{Dismantlability} \label{dismantlability}

\subsection{Preliminaries} \label{dism_prelim}
Let $\BA = \langle A;R_1(\BA),\dots,R_m(\BA)\rangle$ be a
$\sigma$-structure. For $x, y \in A$ we say that $y$ {\em
dominates $x$ in $\BA$}, if for every $R_i \in \sigma$, $j \in
\{1, \ldots, r_i\}$ and $(x_1, \ldots, x_{r_i}) \in R_i(\BA)$ with
$x_j = x$ we also have $(y_1, \ldots, y_{r_i}) \in R_i(\BA)$ with
$y_j = y$ and $y_k = x_k$ for all $k \neq j$. For instance, if
$R_i$ is ternary and $(x,t,x) \in R_i(\BA)$, then for $y$ to
dominate $x$ we must have $(y,t,x) \in R_i(\BA)$ and $(x,t,y) \in
R_i(\BA)$, each of which also implies $(y,t,y) \in R_i(\BA)$. We
say that $x$ {\em is dominated in $\BA$} if it is dominated by
some element $y \in A\setminus \{x\}$. We say that $\BA$ {\em
dismantles to} its induced substructure $\BB$ if there exists a
sequence $x_1,\dots,x_k$ of distinct elements of $A$ such that
$A\setminus B = \{x_1,\dots,x_k\}$ and for each $1 \leq i \leq k$
the element $x_i$ is dominated in the structure induced by $B \cup
\{x_i,\dots,x_k\}$. In other words, the structure $\BB$ can be
obtained from $\BA$ by successively removing dominated elements;
the sequence $x_1,\dots,x_k$ is then called a {\em dismantling
sequence}. Note that if $\BA_x$ is the substructure of $\BA$
induced by $A\setminus\{x\}$, where $x$ is dominated by $y$ in
$\BA$, then we can define a retraction $\rho: \BA \rightarrow
\BA_x$ by putting $\rho(x) = y$ and $\rho(z) = z$ for all $z \neq
x$. Using composition we then see that if $\BA$ dismantles to
$\BB$ then $\BB$ is a retract of $\BA$ (the converse does not hold
in general). Our first result shows that ``dismantling $\BA$ to
$\BB$'' can be done greedily.
\begin{lemma} \label{greedy}
Let $\BA, \BB$ be $\sigma$-structures such that $\BA$ dismantles
to $\BB$. Then for every dominated element $a \in A \setminus B$
of $\BA$, the substructure $\BA_a$ of $\BA$ induced by $A
\setminus \{a\}$ dismantles to $\BB$.
\end{lemma}
\begin{proof}
Let $x_1, \ldots, x_k$ be a dismantling sequence of $\BA$ on
$\BB$. Note that for some index $j$ we have $x_j = a$. We will
show that by removing $x_j$ and perhaps rearranging the sequence
we get a dismantling sequence of $\BA_a$ on $\BB$. For $i = 1,
\ldots, k$ let $y_i$ be an element dominating $x_i$ in the
substructure $\BA_i$ of $\BA$ induced by $B \cup \{x_i, \ldots,
x_k\}$. Note that for some indices $i$ there may be many choices
for $y_i$, and whenever $y_i \neq a$, $y_i$ also dominates $x_i$
in the substructure of $\BA_a$ induced by $B \cup \{x_i, \ldots,
x_k\} \setminus \{x_j\}$. Thus it suffices to show that for all $i
\in \{1, \ldots, k\}$, we can select $y_i$ other than $a$.

Let $i$ be the smallest index such that $y_i = a$, and let $b$ be
an element dominating $a$ in $\BA$. Note that if $b \not \in
\{x_1, \ldots, x_{i-1}\}$, then $b$ also dominates $x_i$ in the
substructure of $\BA_a$ induced by $B \cup \{x_i, \ldots, x_k\}
\setminus \{x_j\}$, hence we can select $y_i = b$ instead. Thus we
can assume that $b = x_{i'}$ for some $i' < i$. We then define a
finite increasing sequence $i_0, i_1, \ldots, i_{\ell}$ by putting
$i_0 = i'$, and letting $i_{p+1}$ be the index in $\{i_p +1,
\ldots, i - 1\}$ such that $y_{i_p} = x_{i_{p+1}}$ if such an
index exists. Then $x_i$ is dominated by $a$ in $\BA_i$, which is
dominated by $b = x_{i_0}$ in $\BA$. For $p = 0, \ldots, \ell-1$,
$x_{i_p}$ is dominated by $y_{i_p} = x_{i_{p+1}}$ in $\BA_{i_p}$,
and $x_{i_{\ell}}$ is dominated by $y_{i_{\ell}} \neq a$ in
$\BA_{i_{\ell}}$. If $y_{i_{\ell}} \neq x_i$, then $y_{i_{\ell}}$
also dominates $x_i$ in $\BA_i$ hence we can select $y_i =
y_{i_{\ell}}$ instead of $y_i = a$. If $y_{i_{\ell}} = x_i$, then
$x_i$ and $a = x_j$ dominate each other in $\BA_i$. In this case,
$x_1, \ldots, x_{i-1}$ is a dismantling sequence of $\BA_a$ on its
substructure induced by $B \cup \{x_i, \ldots, x_{j-1} \} \cup
\{x_{j+1}, \ldots, x_k\}$, which is isomorphic to $\BA_{i+1}$ via
an isomorphism which fixes $B$, whence $\BA_a$ dismantles to
$\BB$.
\end{proof}

\subsection{Exponentiation}
Let $\BA$ and $\BB$ be two $\sigma$-structures. The {\em $\BA$-th
power of $\BB$} is the $\sigma$-structure
$$\BB^\BA = \langle B^A; R_1(\BB^\BA), \ldots, R_m(\BB^\BA) \rangle,$$
where $B^A$ is the set of all maps from $A$ to $B$, and for $i =
1, \ldots, m$ the relation $R_i(\BB^\BA)$ consists of all
hyperedges $(f_1,\dots,f_{r_i})$ such that
$(f_1(x_1),\dots,f_{r_i}(x_{r_i})) \in R_i(\BB)$ whenever
$(x_1,\dots,x_{r_i}) \in R_i(\BA)$. This definition is derived
from the following correspondence, whose proof is straightforward.
\begin{lemma} \label{expoprod}
Let $\phi: \BA \times \BC \rightarrow \BB$ be a homomorphism. Then
the map $\psi: C \rightarrow B^A$ defined by $\psi(c) = f_c$,
where $f_c(a) = \phi(a,c)$, is a homomorphism from $\BC$ to
$\BB^\BA$. Conversely, if $\psi: \BC \rightarrow \BB^\BA$ is a
homomorphism, then the map $\phi: A \times C \rightarrow B$
defined by $\phi(a,c) = \phi(c)(a)$ is a homomorphism from $\BA
\times \BC$ to $\BB$.
\end{lemma}
In particular the homomorphisms from $\BA$ to itself can be viewed
as homomorphisms from the product of $\BA$ and a loop to $\BA$,
which then correspond to loops in $\BA^\BA$.

Now suppose that $a$ is dominated by $b$ in $\BA$, and let $\rho$
be the retraction which maps $a$ to $b$ and fixes every other
element of $A$. Then, considered as an element of $\BA^\BA$,
$\rho$ is a ``neighbour'' of the identity in the sense that there
exists a homomorphism $\psi$ from the $1$-link $\BL_1$ to
$\BA^\BA$ defined by $\psi(0) = \mbox{id}_A$ and $\psi(1) = \rho$.
The main result of this section is a generalisation of this
observation to the dismantling process in general.
\begin{lemma} \label{dismantle} Let $\BA$ be a $\sigma$-structure and
let $\BB$ be a substructure of $\BA$. Then $\BA$  dismantles to
$\BB$ if and only if there exist some $n \geq 0$ and a
homomorphism $P:\BL_n \rightarrow \BA^\BA$ such that
\begin{itemize}
\item[(i)] $P(0) = \mbox{id}_A$, \item[(ii)] $B$ is fixed
pointwise by $P(t)$ for every $t=0,\dots,n$, \item[(iii)] $P(n)$
is a retraction onto $B$.
\end{itemize}
\end{lemma}
We call two homomorphisms $f, g: \BA \rightarrow \BA$ {\em
adjacent} if there is a homomorphism $P$ from $\BL_1$ to $\BA^\BA$
such that $P(0) = f$ and $P(1) = g$. Hence Lemma~\ref{dismantle}
states that $\BA$ dismantles to $\BB$ if and only if there is a
link of homomorphisms fixing $B$ pointwise which joins the
identity on $A$ to a retraction onto $B$. The proof will use the
following property of composition in powers, whose proof is a
straightforward application of the definition.
\begin{lemma} \label{composition}
Let $\BA, \BB, \BC$ be $\sigma$-structures. Then the map $\phi:
\BA^\BB \times \BB^\BC \rightarrow \BA^\BC$ defined by $\phi(f,g)
= f \circ g$ is a homomorphism. In particular for any integer $p$,
the map $\varepsilon_p: \BA^\BA \rightarrow \BA^\BA$ defined by
$\varepsilon_p(f) = f^p$ is a homomorphism.
\end{lemma}
For every $f \in A^A$, and $a \in A$, there exist integers $0 \leq
i < j \leq \vert A \vert$ such that we have $f^j(a) = f^i(a)$; we
say that $a$ has {\em finite period under $f$} if we can take $i =
0$. For $p = \vert A \vert !$, we then have $f^p(a) = a$ if $a$
has finite period under $f$, and otherwise $f^p(a)$ has finite
period under $f$. Thus $f^p$ is a set-theoretic retraction of $A$
onto the set of its elements of finite period under $f$. Therefore
for $p = \vert A \vert !$, the homomorphism $\varepsilon_p$
defined in Lemma~\ref{composition} is a retraction of $\BA^\BA$
onto its substructure induced by the set-theoretic retractions of
$A$.

\begin{proof} [Proof of Lemma \ref{dismantle}]
Suppose that $\BA$ dismantles to $\BB$, and let $x_1, \ldots, x_k$
be a dismantling sequence of $\BA$ on $\BB$. For $t = 1, \ldots,
k$, let $y_t \neq x_t$ be an element dominating $x_t$
in the substructure
of $\BA$ induced by $B \cup \{x_t, \ldots, x_k\}$. We define a
sequence $\rho_0, \rho_1, \ldots, \rho_k$ of retractions
inductively by $\rho_0 = \mbox{id}_A$, $\rho_t(z) = y_t$ if
$\rho_{t-1}(z) = x_t$ and $\rho_t(z) = \rho_{t-1}(z)$ otherwise.
Let $P: \BL_k \rightarrow \BA^\BA$ be defined by $P(t) = \rho_t$.
Then $P(0)$ is the identity, $B$ is fixed by each $P(t)$, and
$P(k)$ is a retraction onto $B$. We show that $P$ is a
homomorphism.

For $R_i \in \sigma$, let $(t_1, \ldots, t_{r_i})$ be an element
of $R_i(\BL_k)$. Then there exists an index $t \in \{1, \ldots,
k\}$ and a subset $J$ of $\{1, \ldots, r_i\}$ such that $t_j = t$
if $j \in J$ and $t_j = t-1$ otherwise. We then have $(P(t_1),
\ldots, P(t_{r_i})) = (f_1, \ldots, f_{r_i})$ where $f_j = \rho_t$
if $j \in J$ and $f_j = \rho_{t-1}$ otherwise. For every $(a_1,
\ldots, a_{r_i}) \in R_i(\BA)$, we have $(\rho_{t-1}(a_1), \ldots,
\rho_{t-1}(a_{r_i}))\in R_i(\BA)$, since $\rho_{t-1}$ is a
homomorphism. Now $(f_1(a_1), \ldots, f_{r_i}(a_{r_i}))$ coincides
with $(\rho_{t-1}(a_1), \ldots, \rho_{t-1}(a_{r_i}))$ except for
some possible coordinates in $J$ where $y_t$ replaces $x_t$. Since
$\{\rho_{t-1}(a_1), \ldots, \rho_{t-1}(a_{r_i})\} \subseteq B \cup
\{ x_t \ldots, x_k\}$ and $y_t$ dominates $x_t$ in the
substructure of $\BA$ induced by that subset, we then have
$(f_1(a_1), \ldots, f_{r_i}(a_{r_i})) \in R_i(\BA)$. Thus $(f_1,
\ldots, f_{r_i}) \in R_i(\BA^\BA)$. This shows that $P$ is a
homomorphism.

Conversely, suppose that $P: \BL_n \rightarrow \BA^\BA$ is a
homomorphism such that for $\phi_t = P(t), t = 0, \ldots, n$ we
have $\phi_0 = \mbox{id}_A$, $B$ is fixed pointwise by each
$\phi_t$ and $\phi_n$ is a retraction onto $B$. Put $p = \vert A
\vert !$. We define three maps as follows.
\begin{itemize}
\item[(i)] $P': \BL_n \rightarrow \BA^\BA$ is defined by $P'(t) =
\rho_t := \phi_t^p$. Thus $P' = \varepsilon_p \circ P$, which is a
homomorphism by Lemma \ref{composition}.
\item[(ii)] $P'': \BL_n
\rightarrow \BA^\BA$, where $P''(t) = \psi_t$ is defined
recursively by $\psi_0 = \rho_0$ and $\psi_t = \psi_{t-1} \circ
\rho_t$ for $t = 1, \ldots, n$. Since $\rho_t$ is idempotent,
$\psi_t = \psi_{t-1} \circ \rho_t = \psi_{t-1} \circ \rho_t \circ \rho_t
= \psi_t \circ \rho_t$ is adjacent to $\psi_{t}
\circ \rho_{t+1} = \psi_{t+1}$ by  Lemma \ref{composition},
whence $P''$ is a homomorphism.
\item[(iii)] $P''' = \varepsilon_p \circ P'': \BL_n \rightarrow
\BA^\BA$ is a homomorphism by Lemma \ref{composition}.
\end{itemize}
Note that $P'''(0) = P''(0) = P'(0) = P(0) = \mbox{id}_A$, and
since every $P(t)$ fixes $B$, $P'''(n) = P''(n) = P'(n) = P(n)$
which is a retraction onto $B$. Also, for $t = 1, \ldots, n$,
$\hat{\rho}_t := P'''(t)$ is a retraction whose image
$\mbox{im}(\hat{\rho}_t)$ is contained in that of
$\hat{\rho}_{t-1}$. We can then show that every $a \in
\mbox{im}(\hat{\rho}_{t-1}) \setminus \mbox{im}(\hat{\rho}_t)$ is
dominated by $\hat{\rho}_t(a)$ in the substructure $\BA_{t-1}$ of
$\BA$ induced by $\mbox{im}(\hat{\rho}_{t-1})$. Indeed, for $R_i
\in \sigma$ and $(a_1, \ldots, a_{r_i}) \in R_i(\BA_{t-1})$ such
that $a_j = a$ for some index $j$, we have that
$(\hat{\rho}_{t-1}(a_1),\ldots, \hat{\rho}_{t-1}(a_{j-1}),
\hat{\rho}_{t}(a_{j}),
\hat{\rho}_{t-1}(a_{j+1}),\ldots,\hat{\rho}_{t-1}(a_{r_i}))$ is in
$R_i(\BA)$ since $\hat{\rho}_{t}$ is adjacent to
$\hat{\rho}_{t-1}$, whence $\hat{\rho}_{t}(a)$ dominates $a$ in
$\BA_{t-1}$. Therefore $\BA$ dismantles to its substructure
induced by $\hat{\rho}_n(A) = B$.
\end{proof}

\subsection{$\BA^{\left (\BA^2\right )}$ and
$\left (\BA^2 \right )^{\left (\BA^2\right )}$}

Here we interpret Lemma \ref{lxa2toa} in terms of
exponential structures. For a $\sigma$-structure $\BA$ we denote
$\pi_1$ and $\pi_2$ the two projections of $\BA^2$ on $\BA$. The
{\em diagonal} of $\BA^2$ is its substructure $\Delta_{\BA^2}$
induced by $\{ (a,a) : a \in A\}$.
\begin{lemma} \label{aalaa2}
Let $\BA$ be a $\sigma$-structure and $n$ an integer. If there
exists a homomorphism $P: \BL_n \rightarrow \BA^{\left
(\BA^2\right )}$ such that $P(0) = \pi_1$ and $P(n) = \pi_2$, then
there exists a homomorphism from  $\BL_n \times \BA^2 /\!\!\sim_n$
to $\BA$. If $\BA$ is a core, the converse also holds.
\end{lemma}
\begin{lemma} \label{a2alaa2}
Let $\BA$ be a $\sigma$-structure. If $\BA^2$ dismantles to its
diagonal, then for some $n$ there exists a homomorphism $P: \BL_n
\rightarrow \BA^{\left (\BA^2\right )}$ such that $P(0) = \pi_1$
and $P(n) = \pi_2$. If $\BA$ is a core, the converse also holds.
\end{lemma}

\begin{proof} [Proof of Lemma \ref{aalaa2}]
By Lemma \ref{expoprod} a homomorphism $P: \BL_n \rightarrow
\BA^{\left (\BA^2\right )}$ corresponds to the homomorphism $\phi:
\BL_n \times \BA^2 \rightarrow \BA$ defined by $\phi(i,a,b) =
P(i)(a,b)$. If $P(0) = \pi_1$ and $P(n) = \pi_2$, then $\phi$ is
constant on every $\sim_n$-equivalence class, hence we can define
a homomorphism $\psi: \BL_n \times \BA^2 /\!\!\sim_n \rightarrow
\BA$ by $\psi((t,a,b)/\!\!\sim_n) = \phi(t,a,b)$.

Conversely, any homomorphism $\psi: \BL_n \times \BA^2 /\!\!\sim_n
\rightarrow \BA$ can be composed with the quotient map $q: \BL_n
\times \BA^2  \rightarrow \BL_n \times \BA^2 /\!\!\sim_n$ to give
a homomorphism $q \circ \psi: \BL_n \times \BA^2  \rightarrow
\BA$. By Lemma \ref{expoprod}, $q \circ \psi$ corresponds to a
homomorphism $P: \BL_n \rightarrow \BA^{\left (\BA^2\right )}$,
and by definition of $\sim_n$ there exist homomorphisms $\phi_1,
\phi_2$ from $\BA$ to itself such that $P(0) = \phi_1 \circ \pi_1$
and $P(n) = \phi_2 \circ \pi_2$. If $\BA$ is a core, then by
Lemma~\ref{coreisrigid}, $\phi_1$ and $\phi_2$ are both the
identity, whence $P(0) = \pi_1$ and $P(n) = \pi_2$.
\end{proof}

\begin{proof} [Proof of of Lemma \ref{a2alaa2}]
Suppose that $\BA^2$ dismantles to its diagonal $\Delta_{\BA^2}$.
By Lemma~\ref{dismantle}, for some $n$ there exists a homomorphism
$P : \BL_n \rightarrow \left (\BA^2\right )^{\left (\BA^2\right
)}$ such that $P(0)$ is the identity and $P(n)$ is a retraction on
$\Delta_{\BA^2}$. We can then define a homomorphism $P': \BL_{2n}
\rightarrow \BA^{\left (\BA^2\right )}$ by $P'(t) = \pi_1 \circ
P(t)$ and $P'(2n - t) = \pi_2 \circ P(t)$ for $t = 0, \ldots, n$.
Indeed both definitions of $P'(n)$ coincide since $P(n)$ is a
retraction on $\Delta_{\BA^2}$, and since $P(0)$ is the identity,
$P'(0) = \pi_1$ and $P'(2n) = \pi_2$.

Conversely, for every homomorphism $P: \BL_{n} \rightarrow
\BA^{\left (\BA^2\right )}$ such that $P(0) = \pi_1$ and $P(n) =
\pi_2$, we can define a homomorphism $P' : \BL_n \rightarrow \left
(\BA^2\right )^{\left (\BA^2\right )}$ by $P'(t) = (P(t),\pi_2)$.
Then $P'(0) = (\pi_1, \pi_2)$ is the identity and $P'(n) =
(\pi_2,\pi_2)$ is a retraction on $\Delta_{\BA^2}$. If $\BA$ is a
core, then since $\Delta_{\BA^2}$ is isomorphic to $\BA$ via the
canonical isomorphism, the restriction of every $P(t)$ to
$\Delta_{\BA^2}$ must coincide with $\pi_2$ by
Lemma~\ref{coreisrigid}. Hence for $t = 0, \ldots, n$, $P'(t)$
fixes $\Delta_{\BA^2}$. Therefore $\BA^2$ dismantles to $\BA$ by
Lemma~\ref{dismantle}.
\end{proof}

Let $\BA$ be a relational structure such that $\BA^2$ dismantles
to $\Delta_{\BA^2}$. Then by Lemma~\ref{a2alaa2}, $\BA^{\left
(\BA^2\right )}$ contains a link between the two projections, thus
for some $n$ there exists a homomorphism from $\BL_n \times \BA^2
/\!\!\sim_n$ to $\BA$ by Lemma~\ref{aalaa2}. Hence, by
Theorem~\ref{focharacterisations}, $\BA$-CSP is first-order
definable. The converse does not hold in general. However, for any
retract $\BB$ of $\BA$, $\BA$-CSP is equivalent to $\BB$-CSP. In
particular, if $\BB$ is the core of $\BA$ and $\BA$-CSP is
first-order definable, then Theorem~\ref{focharacterisations},
Lemma~\ref{aalaa2} and Lemma~\ref{a2alaa2} imply that $\BB^2$
dismantles to $\Delta_{\BB^2}$. Therefore we have proved the
following:
\begin{theorem} \label{foinnp}
A relational structure has a first-order definable CSP if and only
if it has a retract whose square dismantles to its diagonal.
\end{theorem}

\section{The complexity of recognising first-order
definable CSP's} \label{complexsection}

\begin{theorem} \label{foisnpc}
The problem of determining whether a relational structure $\BA$
has a first-order definable CSP is {\bf NP}-complete.
\end{theorem}
In fact, we will show the problem to be {\bf NP}-complete even in
the restricted case of directed graphs. We contrast this with the
following result:
\begin{theorem} \label{focoreispoly}
The problem of determining whether a relational structure $\BA$ is
a core with a first-order definable CSP can be solved in
polynomial time.
\end{theorem}

In particular, Theorem~\ref{focoreispoly} implies that deciding
whether an input core structure $\BA$ has a first-order definable
CSP can be done in polynomial time, but our algorithm does not require
a certificate that the input is a core.

\begin{proof} [Proof of Theorem \ref{foisnpc}]
Theorem \ref{foinnp} shows that the problem is in {\bf NP}. We
will show that 3-SAT reduces to the problem of determining whether
a given digraph has first-order definable CSP. Let ${\mathcal I} =
\bigwedge_{i=0}^{n-1} \left ( L_{i,1} \vee L_{i,2} \vee L_{i,3}
\right )$ be an instance of 3-SAT, where each literal is one of
the variables $x_1, \ldots, x_m$ or its negation, and (without
loss of generality) $L_{i,j} \neq L_{i,j'}$ when $j \neq j'$. We
construct a digraph $H$ such that ${\mathcal I}$ is satisfiable if
and only if $H$ has first-order definable CSP. The vertex-set of
$H$ is $\{0, \ldots, n-1\} \times \{1, 2, 3\}$, and there is an
arc from $(i,j)$ to $(i',j')$ if and only if $i < i'$ and
$L_{i,j}$ is not the negation of $L_{i',j'}$.

Thus the map $\phi$ from $H$ to the transitive tournament $T_n$ on
$n$ vertices defined by $\phi(i,j) = i$ is a homomorphism.
Furthermore it is not hard to see that for every tree $A$ which
admits a homomorphism $\psi: A \mapsto T_n$, there exists a
homomorphism $\hat{\psi}: A \mapsto H$ such that $\psi = \phi
\circ \hat{\psi}$. Thus the trees that map to $H$ are precisely
those which map to $T_n$. Since $T_n$ has finite duality
\cite{nestar}, this means that $H$ has first-order definable CSP
if and only if $T_n$ is the core of $H$ by
Theorem~\ref{foequivalences}.

If ${\mathcal I}$ is satisfiable, then selecting for each $i$ an
index $j_i$ such that $L_{i,j_i}$ is true yields a homomorphic
image $\{ (i,j_i) : 1 \leq i \leq n \}$ of $T_n$ in $H$.
Conversely, if $\{ (i,j_i) : 1 \leq i \leq n \}$ is a homomorphic
image of $T_n$ in $H$, then we can consistently deem the literals
$L_{i,j_i}$ to be true to find a satisfactory truth assignment of
${\mathcal I}$. Therefore ${\mathcal I}$ is satisfiable if and
only if $H$ has first-order definable CSP.
\end{proof}

\begin{proof} [Proof of Theorem \ref{focoreispoly}]
We first test whether $\BA^2$ dismantles to $\Delta_{\BA^2}$.
According to Lemma~\ref{greedy} this step can be performed in
polynomial time using the greedy algorithm. If the answer is
negative, then either $\BA$ is not a core, or it is a core which
does not have a first-order definable CSP. In any case, we output
``no'' and stop. If the answer is positive, then $\BA$ does have
first-order definable CSP, but it may not be a core. For each pair
$a, b \in A, a \neq b$ we form the quotient $\BA_{\{a,b\}}$ of
$\BA$ under the equivalence which identifies $a$ and $b$. By
Theorem~\ref{foequivalences}, $\BA$ has tree duality, hence the
polynomial consistency-check algorithm (see \cite{fedvar2})
detects whether $\BA_{\{a,b\}}$ admits a homomorphism to $\BA$. If
such a homomorphism $\phi$ exists, then $\BA$ admits a
homomorphism to its proper substructure $\phi(\BA_{\{a,b\}})$
hence it is not a core; we then output ``no'' and stop. If no
homomorphism exists from any quotient $\BA_{\{a,b\}}$ to $\BA$,
then $\BA$ is a core. We then output ``yes''.
\end{proof}

\section{Producing solutions of first-order definable CSP's} \label{psfocsp}

Let $\BA$ be a structure such that $\BA^2$ dismantles to its
diagonal. Then $\BA$ has a first-order definable CSP; furthermore
without loss of generality, we can assume that $\BA$ is a core,
since adding to the type $\sigma$ a unary relation for each
element $A$ preserves the dismantling of
$\BA^2$ to its diagonal. Thus, the hyperedge consistency check
algorithm is sufficient to determine whether a structure $\BB$
admits a homomorphism to $\BA$. It is possible to find an explicit
homomorphism from $\BB$ to $\BA$ in polynomial time using vertex
identifications on a trial and error basis. In this section, we
provide an alternative algorithm based on dismantlings of $\BB
\times \BA$. We will use the following variation of Lemma
\ref{greedy}:

\begin{lem} \label{udt}
Let $\BB$ be a structure which dismantles to two substructures
$\BC$ and $\BC'$. If neither of $\BC$ and $\BC'$ have dominated
elements, then $\BC$ and $\BC'$ are isomorphic.
\end{lem}

\proof Let $x_1, \ldots, x_n$ be a dismantling sequence of $\BB$
on $\BC'$. For $i = 1, \ldots, n$, let $\BB_i$ be the substructure
of $\BB$ induced by $\{x_i, \ldots, x_n\} \cup C'$. We define a
sequence $\BC_0 = \BC, \BC_1, \ldots, \BC_n$ of structures such
that  $\BC_i = \BC_{i-1}$ if $x_i$ is not in the universe of
$\BC_{i-1}$, and otherwise $\BC_i$ is obtained from $\BC_{i-1}$ by
replacing the element $x_i$ by an element which dominates it in
$\BB_i$. By induction we prove that
\begin{quote}
 $\BB$ dismantles to $\BC_i$, which is isomorphic to $\BC$.
\end{quote}
Indeed, for $i = 0$ this is given, and the induction step clearly
works when $\BC_i = \BC_{i-1}$. Suppose that $\BC_i$ is obtained
from $\BC_{i-1}$ by replacing $x_i$ by $y$. Note that $y$ is not
already in  $\BC_{i-1}$ since $\BC_{i-1}$ is isomorphic to $\BC$
which contains no dominated elements. By Lemma \ref{greedy}, there
exists a dismantling sequence $z_1, \ldots z_m$ of $\BB_i$ on
$\BC_{i-1}$. By replacing $y$ by $x_i$ in this sequence, we get a
dismantling sequence of $\BB_i$ on $\BC_{i}$, whence $\BB$
dismantles to $\BC_{i}$. Moreover, $\BC_{i-1} \cup \BC_{i}$
clearly dismantles to both $\BC_{i-1}$ and  $\BC_{i}$, whence
$x_i$ and $y$ dominate each other in  $\BC_{i-1} \cup \BC_{i}$.
Therefore $\BC_{i-1}$ and  $\BC_{i}$ are isomorphic.

Thus, $\BC'$ contains a substructure isomorphic to $\BC$. By
interchanging the roles of $\BC$ and $\BC'$, we conclude that
$\BC$ and $\BC'$ are isomorphic. \qed

In a product $\BB \times \BA$,
an element $(b,a)$ is said to be {\em dominated in
the second coordinate} if it is dominated by an element of the
form $(b,a')$. We say that  $\BB \times \BA$ {\em dismantles in
the second coordinate} to its substructure $\BC$ if $\BC$ can be
obtained from $\BB \times \BA$ by successively removing elements
that are dominated in the second coordinate. Note that
dismantlings of $\BB \times \BA$ in the second coordinate can be
considered as ordinary dismantlings, by adding to the type
$\sigma$ one unary relation $R_b = \{ (b,a) : a \in A\}$ for each
$b \in B$. Hence the results of Lemma \ref{udt} apply, and $\BB
\times \BA$ dismantles in the second coordinate to a structure
$\BC$ with no elements dominated in the second coordinate. Such a
structure $\BC$ is unique up to isomorphism. For each $b \in B$,
there exists at least one $a \in A$ such that $(b,a) \in C$. If
for each $b \in B$, there exists exactly one $a = \phi(b)
\in A$ such that $(b,a) \in C$, then $\BC$ is the {\em graph} of
the function $\phi: B \rightarrow A$. The latter is a homomorphism
from $\BB$ to $\BA$ precisely when $\BC$ is isomorphic to $\BB$.

\begin{thm} \label{distohom}
Let $\BA$ be a structure such that $\BA^2$ dismantles to its
diagonal. For a structure $\BB$, let $\BC$ be a substructure of
$\BB \times \BA$ obtained by dismantling in the second coordinate
until no more elements are dominated in the second coordinate.
Then $\BB$ admits a homomorphism to $\BA$ if and only if $\BC$ is
the graph of a homomorphism from $\BB$ to $\BA$.
\end{thm}

\proof Obviously, if $\BB \times \BA$ dismantles in the second
coordinate to the graph of a homomorphism $\psi$ from $\BB$ to
$\BA$, then $\BB$ admits a homomorphism to $\BA$. The proof of the
converse parallels that of Theorem \ref{foinnp}. As mentioned in
the beginning of this section, we can assume that $\BA$ is a core.
Let $\BB$ be a structure which admits a homomorphism to $\BA$.

We first define the structure $(\BL_n \times \BB \times \BA
/\!\!\sim_n)^*$ as follows: $\BL_n \times \BB \times \BA
/\!\!\sim_n$ is the quotient of the product $ \BL_n \times \BB
\times \BA$ under the equivalence $\sim_n$ defined by
$$
(k,b,a) \sim_n (k',b',a') \equiv \left \{
\begin{array}{l}
\mbox{$(k,b,a) = (k',b',a')$} \\
\mbox{or $k = k' = 0$ and $a = a'$}\\
\mbox{or $k = k' = n$ and $b = b'$}.
\end{array}
\right.
$$
Note that since $\BB$ admits a homomorphism to $\BA$, the fiber
$(\{n\} \times \BB \times \BA/\!\!\sim_n)$ is isomorphic to $\BB$,
while $(\{0\} \times \BB \times \BA)/\!\!\sim_n$ is not
necessarily ismomorphic to $\BA$. We complete the structure of
$(\BL_n \times \BB \times \BA /\!\!\sim_n)^*$ by adding a copy of
$\BA$ to the fiber $\{0\} \times \BB \times \BA$: for each $R \in
\sigma$ and $(a_1, \ldots, a_r) \in R(\BA)$, we put
$((0,a_1,b)/\!\!\sim_n, \ldots, (0,a_1,b)/\!\!\sim_n) \in R((\BL_n
\times \BB \times \BA /\!\!\sim_n)^*)$.

As in Lemma \ref{cutends}, the substructures of $(\BL_n \times \BB
\times \BA /\!\!\sim_n)^*$ induced by $\{ (k,a,b)/\!\!\sim_n : k
\neq 0\}$ and $\{ (k,a,b)/\!\!\sim_n : k \neq n\}$ admit natural
homomorphisms to $\BA$ and $\BB$ respectively, whence both of
these admit homomorphisms to $\BA$. Thus if $n$ is larger than the
diameter of the minimal obstructions of $\BA$, then there exists a
homomorphism
$$\alpha:
(\BL_n \times \BB \times \BA /\!\!\sim_n)^* \rightarrow \BA.$$
Note that $\alpha$ corresponds to a link of homomorphisms
$\alpha_k \in \BA^{\BB \times \BA}, k = 0, \ldots, n$, where
$\alpha_0 = \pi_A$ and $\alpha_n = \phi \circ \pi_B$ for some
homomorphism $\phi: \BB \rightarrow \BA$. We use $\alpha$ to
define a link of homomorphisms $\beta_k \in {(\BB \times
\BA)}^{\BB \times \BA}, k = 0, \ldots, n$ by
$$
\beta_k(b,a) = (b,\alpha((k,b,a)/\!\!\sim_n ).
$$
Thus, $\beta_0 = \mbox{id}_{\BB \times \BA}$, $\beta_k(b,a) = (b,
\beta_k'(b,a)), k = 1, \ldots, n-1$ and $\beta_n(b,a) = (b,
\phi(b))$. There are two desirable properties which would allow us
to reach our conclusion: If $\beta_0, \ldots \beta_n$ were a link
of retractions such that $\beta_0(A \times B) \supseteq \beta_1(A
\times B) \supseteq \ldots \supseteq \beta_n(A \times B)$, then by
Lemma \ref{dismantle} we would have that $\BB \times \BA$
dismantles on $\beta_n(\BB \times \BA)$. However the current link
$\beta_0, \ldots \beta_n$ may have neither of these properties.
Thus we will repeatedly modify our link through the following two
procedures:
\begin{itemize}
\item[(i)] If $\gamma_0, \ldots, \gamma_n \in {(\BB \times
\BA)}^{\BB \times \BA}$ is a link with the same properties as
$\beta_0, \ldots \beta_n$ above, then for $p = \vert A\vert !$ the
functions  $\rho_0, \ldots, \rho_n \in {(\BB \times \BA)}^{\BB
\times \BA}$ defined by $\rho_k = \gamma_k^p$ form a link of
retractions, where $\rho_0 = \mbox{id}_{\BB \times \BA}$,
$\rho_k(b,a) = (b, \rho_k'(b,a)), k = 1, \ldots, n-1$ and
$\beta_n(b,a) = (b, \rho_n'(b))$ for some $\rho_n': B \rightarrow
A$. However we do not necessarily have $\rho_0(A \times B)
\supseteq \rho_1(A \times B) \supseteq \ldots \supseteq \rho_n(A
\times B)$.
 \item[(ii)] If  $\rho_0, \ldots, \rho_n \in {(\BB
\times \BA)}^{\BB \times \BA}$ is a link of retractions with the
properties given in (i), we define the sequence $\gamma_0, \ldots,
\gamma_n \in {(\BB \times \BA)}^{\BB \times \BA}$ recursively by
$\gamma_0 = \rho_0$ and $\gamma_{k} = \gamma_{k-1} \circ \rho_k, k
= 1, \ldots n$. Then we clearly have $\gamma_0(B \times A)
\supseteq \gamma_1(B \times A) \supseteq \ldots \gamma_n(B \times
A)$. Moreover, $\gamma_0 = \mbox{id}_{\BB \times \BA}$,
$\gamma_k(b,a) = (b, \gamma_k'(b,a)), k = 1, \ldots, n-1$ and
$\gamma_n(b,a) = (b, \gamma_n'(b))$ for some $\gamma_n': B
\rightarrow A$. In particular, $\gamma_0 = \rho_0 = \mbox{id}_{\BB
\times \BA}$ and $\gamma_1 = \rho_1$ are adjacent, and if
$\gamma_{k-1}$ and $\gamma_k$ are adjacent, then so are $\gamma_k
= \gamma_{k-1} \circ \rho_k$ and $\gamma_{k+1} = \gamma_{k} \circ
\rho_{k+1}$ by Lemma \ref{composition}. Thus, $\gamma_0, \ldots,
\gamma_n$ forms a link, though these homomorphisms may not be
retractions.
\end{itemize}
After an initial run through steps (i) and (ii), every time we
need to repeat step (ii) it is because the previous step (i)
reduced the size of the images of some functions in the link. Thus
after some repetitions, we eventually get a link $\delta_0,
\ldots, \delta_n \in {(\BB \times \BA)}^{\BB \times \BA}$ such
that $\delta_0 = \mbox{id}_{\BB \times \BA}$, each $\delta_k$ is a
retraction and $\delta_0(A \times B) \supseteq \delta_1(A \times B)
\supseteq \ldots \supseteq \delta_n(A \times B)$. Thus the map $P:
\BL_n \rightarrow {(\BB \times \BA)}^{\BB \times \BA}$ defined by
$P(k) = \delta_k$ satisfies the hypotheses of Lemma
\ref{dismantle} whence $A \times B$ dismantles to $\delta_n(\BB
\times \BA)$. Moreover, each function $\delta_k$ is of the form
$\delta_k(b,a) = (b,\delta_k'(b,a))$, and also preserve the
relations $R_b = \{ (b,a) : A \in A\}, b \in B$. Thus $A \times B$
dismantles in the second coordinate to $\delta_n(\BB \times \BA) =
\{ (b, \psi(b)) : b \in B\}$ for some function $\psi: B
\rightarrow A$. Since there exists a homomorphism $\phi: \BB
\rightarrow \BA$ and the dismantling sequence induces a
homomorphism from the graph of $\phi$ to that of $\psi$, we
conclude that $\psi$ is indeed a homomorphism from $\BB$ to $\BA$.
\qed

Given a structure $\BA$ such that $\BA^2$ dismantles to its
diagonal, Theorem \ref{distohom} provides the following algorithm
for deciding whether a structure $\BB$ admits a homomorphism to
$\BA$: We dismantle $\BB \times \BA$ in the second coordinate
until we get a structure $\BC$ with no dominations in the second
coordinate. We then have the following possibilities:
\begin{itemize}
\item[(i)] If $\BC$ is not a graph, then there is no homomorphism
from $\BB$ to $\BA$. \item[(ii)] If $\BC$ is a graph, $\BC = \{
(b, \phi(b)) : b \in B\}$ where $\phi: B \rightarrow A$ is not a
homomorphism from $\BB$ to $\BA$, then there is no homomorphism
from $\BB$ to $\BA$. \item[(iii)] Otherwise, $\BB$ admits a
homomorphism to $\BA$, and $\BC$ is the graph of such a
homomorphism $\phi: \BB \rightarrow \BA$.
\end{itemize}
This algorithm works a bit like the hyperedge consistency check,
with the list of an element $b$ of $B$ identified with the fiber
$\{ (b,a) : a \in A\}$. In the dismantling algorithm, an element
is removed from a list if it becomes redundant rather than
inconsistent. Both algorithms work in $O(|B|^{d+2})$ time, where
$d$ is the maximum arity in $\sigma$.

\section{Inferred constraints and {\bf L}-complete CSP's}

In this section we analyse the computational complexity of CSP's
whose basic relations are inferred from those of a first-order
definable CSP. Let $\Gamma$ be a
 set of relations on the finite set $A$. The {\em relational
 clone generated by $\Gamma$}, denoted by $\langle \Gamma \rangle$,
is the set of relations on $A$
 inferred from the relations in $\Gamma$, i.e. definable from relations
in $\Gamma$ via primitive positive
 formulas. We now give equivalent combinatorial and algebraic
descriptions of the
 relations in $\langle \Gamma \rangle$  (see
e.g.~\cite{CohenJ06}). Recall from Section \ref{section_reference}
that an operation $f$ on a set $A$ {\em preserves} a relation
$\theta$ on $A$ if $f$ is a homomorphism from $\BA^n$ to $\BA$
where $\BA = \langle A;\theta\rangle$.

\begin{lemma} \label{6constr} Let $\Gamma$ be a finite set of relations on
$A$ and let $\theta$ be a $k$-ary relation on $A$. Then the
following conditions are equivalent:
\begin{enumerate}
\item $\theta \in \langle \Gamma  \rangle$; \item every operation
on $A$ that preserves every relation in $\Gamma$ also preserves
$\theta$; \item there exists a (primitive positive) formula
$$\phi(x_1,\dots,x_k) \equiv \exists y_1,\dots, \exists y_m
\psi(x_1,\dots,x_k,y_1,\dots,y_m)$$ where $\psi$ is a  conjunction
of atomic formulas with relations in $\Gamma \cup \{=\}$ such that
$(a_1,\dots,a_k)\in \theta$ if and only if $\phi(a_1,\dots,a_k)$
holds;

\item there exists a structure $\BX$ of the same signature as the
structure $\BA = \langle A;\Gamma \rangle$, and elements
$x_1,\dots,x_k \in X$ such that
$$\theta=\{(f(x_1),\dots,f(x_k)): f:\BX \rightarrow \BA\ \mbox{ a homomorphism}\}.$$
\end{enumerate}
\end{lemma}
\qed

 A relation $\theta$ of arity $k \geq 2$ is {\em redundant} if
there exist  indices $i<j$ such that $x_i = x_j$ for any tuple
$\overline{x} \in \theta$; otherwise we say that $\theta$ is {\em
irredundant}. If there exist indices $i<j$ such that $x_i = x_j$
for any tuple $\overline{x} \in \theta$, and furthermore there
exist at least two distinct values $a$ and $b$ that appear as the
$i$-th coordinate of  tuples in $\theta$, then we say that
$\theta$ is {\em biredundant}. Stated differently, $\theta$ is
biredundant if the projection of $\theta$ onto two indices yields
the equality relation on a set with at least 2 elements.

\begin{theorem} Let $\BA$ be a core structure such that $\BA$-CSP is first-order
definable, and let $\BB$ be a structure whose basic relations are
contained in the relational clone generated by the basic relations
of $\BA$. Then \begin{enumerate} \item The problem $\BB$-CSP is in
{\bf L};  \item if $\BB$-CSP is not first-order definable, then it
is {\bf L}-complete; \item if none of the basic relations of $\BB$
is biredundant then $\BB$-CSP is first-order definable; if $\BB$
is a core the converse holds as well.
\end{enumerate}

\end{theorem}

\begin{proof} The first two statements follow from Theorem 5 of
\cite{ELT07} and Theorem 3.1 of \cite{LT07}. Indeed, the problem
$\neg$($\BB$-CSP) is definable in symmetric Datalog, which is
enough to ensure that $\BB$-CSP is solvable in logspace.
Furthermore, every CSP which is not first-order definable is {\bf
L}-hard.

For the third statement we argue as follows: suppose first that no
basic relation of $\BB$ is biredundant. Since $\BA$ is a core with
first-order definable CSP, by Corollary \ref{1antoa}  there exists
a map $f$ which is a homomorphism from $^1\BA^n$ to $\BA$. We
shall prove that $f$ is also a homomorphism from $^1\BB^n$ to
$\BB$ which will conclude the proof by Corollary \ref{1antoa}. Let
$\theta \in \langle \Gamma \rangle$.
If $\theta$
is irredundant then in the description of $\theta$ in Lemma
\ref{6constr} (4) we may choose the elements $x_1,\dots,x_k$ to be
distinct. Let $f_1,\dots,f_{n-1}$ be homomorphisms from $\BX$ to
$\BA$ yielding tuples in $\theta$, and let $h:X \rightarrow A$ be
{\em any} map. It is easy to see  that the map
$p=f(f_1,\dots,f_{n-1},h)$ is a homomorphism from $\BX$ to $\BA$,
and hence $f$ is 1-tolerant for $\theta$. In the case where
$\theta$ is redundant, the argument is almost the same: if for
some indices we have $x_i=x_j$, since $\theta$ is not biredundant,
it follows that the value of $f_1,\dots,f_{n-1}$ at $x_i$ and
$x_j$ is a unique value, call it $a$; since $f$ is a
near-unanimity operation by  Lemma \ref{nuf}, it follows that the
value of $p$ at $x_i$ and $x_j$ is the same and so the tuple
produced by $p$ is in $\theta$.

Conversely, suppose that $\BB$ is a core and that one of its basic
relations  is biredundant: we shall show that the structure
$\BB^2$ does not dismantle to the diagonal. Indeed, suppose that
$\theta$ is biredundant and without loss of generality suppose
that its projection on the first two coordinates is the equality
relation on some subset $B$ of $A$ containing elements 0 and 1.
Suppose  that we have a dismantling of $\BB^2$: let
$A^2=X_0,\dots,X_k$ be the successive subsets of $A^2$ obtained by
removal of single elements. We prove by induction that for every
$i$ there exists a tuple of the form $((0,1),(0,1),\dots) \in
\theta(\BB^2)$ with all entries in $X_i$. This is clear for $i=0$.
Now suppose that there is such a tuple $\overline{x} \in
\theta(\BB^2)$ with all entries in $X_i$ and that $X_{i+1}$ is
obtained from $X_i$ by removal of $(c,d)$. If $(c,d)$ doesn't
appear in $\overline{x}$ then we're done; otherwise by definition
of dismantling there exists some element $(c',d') \in X_{i+1}$
that dominates $(c,d)$ and so the tuple obtained from
$\overline{x}$ by replacing every occurrence of $(c,d)$ by
$(c',d')$ is in $\theta(\BB^2)$. It is clear that $(c,d)\neq(0,1)$
because otherwise the tuple $((0,1),(c',d'),...)$ would be in
$\theta(\BB^2)$ contrary to the fact that $\theta(\BB^2)$ is
biredundant. Hence there is a tuple of the desired form with
entries in $X_{i+1}$, showing that no dismantling can end in the
diagonal.
\end{proof}

\section{Conclusion}

We have described a simple polynomial-time algorithm that
determines if a finite relational structure is a core with
first-order definable CSP (Theorem~\ref{focoreispoly}),
and have proved that deciding FO-definability  is {\bf NP}-complete
(Theorem~\ref{foisnpc}). We have also given various
characterisations of FO-definable structures in terms of sets of
obstructions (Theorem~\ref{foequivalences}), and proved that
core structures with finite duality admit a 1-tolerant near-unanimity
operation (Corollary~\ref{1antoa} and Lemma~\ref{nuf}).

 Feder and Vardi's Theorem~\ref{tddecidable} shows that
the problem of determining whether an input structure $\BA$ has
tree duality is decidable. In fact the proof of
Theorem~\ref{foisnpc} also implies that this problem is NP-hard,
but for the moment it is not known to belong to NP or even to
P-space. It would be interesting to have these issues resolved.

In the case of first-order definable CSP's, we now have an
algorithm which outputs a yes-no answer to the question as to
whether an input structure $\BA$ has a first-order definable CSP.
Using Lemma~\ref{treesofagivendiameter}, it is possible to modify
it so that in the case where $\BA$-CSP is first-order definable,
it outputs a first-order sentence $\Phi_{\BA}$ such that $\BB$
admits a homomorphism to $\BA$ if and only if $\Phi_{\BA}$ is true
on $\BB$. However the upper bound on the length of $\Phi_{\BA}$
involves a tower of exponents. It is not clear whether this is
realistic; \cite{nestar2} reports cases where the length of
$\Phi_{\BA}$ can be logarithmic in terms of the size of $\BA$, but
there are no examples in the direction of the other extreme.

\end{document}